\documentclass[12pt]{article}

\usepackage{setspace}
\usepackage{authblk}
\usepackage{blindtext}
\usepackage[margin=1in]{geometry}
\usepackage[utf8]{inputenc}
\usepackage{amsmath}
\usepackage{amsfonts}
\usepackage{amssymb}
\usepackage{amsthm}
\usepackage{makeidx}
\usepackage{graphicx}
\usepackage{natbib}
\usepackage{hyperref}
\usepackage{abstract}
\usepackage{rotating}
\usepackage{authblk}
\usepackage{multirow}
\usepackage{array}
\usepackage{amsthm}

\usepackage[makeroom]{cancel}
\usepackage{booktabs}
\usepackage{lscape}
\usepackage[]{algorithm2e}
\usepackage{tikz}
\usetikzlibrary{arrows}
\usepackage[colorinlistoftodos]{todonotes}
\usepackage{caption}
\usepackage{subcaption}
\usetikzlibrary{decorations.pathreplacing}
\usepackage{setspace}
\usepackage{tabularx}
\usepackage{ctable}

\usepackage{geometry}

\DeclareMathOperator*{\argmin}{arg\,min}

\newcommand{\nocontentsline}[3]{}
\newcommand{\tocless}[2]{\bgroup\let\addcontentsline=\nocontentsline#1{#2}\egroup}

\newtheorem{proposition}{Proposition}
\newtheorem{lemma}{Lemma}
\newtheorem{problem}{Problem}

\newtheorem{corollary}{Corollary}
\newtheorem{assumption}{Assumption}

\newtheorem*{definition}{Definition}

\title{
\textsc{Compressing over-the-counter markets}\footnote{\tiny We are grateful to Tuomas Peltonen for extensive comments and discussions at all stages. We also thank for comments and suggestions: Darrell Duffie, Per Sjoberg, Mark Yallop, Mathias Dewatripont, Andre Sapir, Rama Cont, Michael Gofman, Hoaxiang Zhu, Stephen O'Connor, Marco Pagano, Diana Iercosan (Discussant), Javier Suarez, Antoine Mandel, Kartik Anand (Discussant), Deborah Lucas, Andrew Lo, Hans Degryse, Mike Mariathasan, Jorge Abad, Inaki Aldasoro, Christoph Aymanns, Bernard Fortz, Daniel Gros, Renaud Lambiotte, Steven Ongena, Marc Chesney, Ashkan Nikeghbali, Paulina Przewoska, Roberto Stok, Olaf Weeken, Benjamin Vandermarliere, Xianglin Flora Meng, Lucio Biaise, Nordine Abidi (Discussant) and participants to the ESRB Joint Expert Group on Interconnectedness meetings, De Nederlandsche Bank research seminar, the ESRB research seminar, the joint ESRB/BoF/RiskLab SRA conference, the LSE Systemic Risk Centre Conference on Systemic Risk Modeling, the BENET 2016 conference, the 2017 AEA/ASSA Chicago meeting, the Macro Connection MIT Media Lab meeting, the NaXys seminar, the Imperial College Mathematical Finance seminar, the WEAI/IBEFA Conference, the MIT Institute for Data, Systems and Society seminar, the SIAM Workshop on Networks, the Second ESRB Annual Conference, the 3rd IWH-FIN-FIRE Workshop, the HEC Montr\'eal Workshop of Measurement and Control of Systemic Risk, the Worcester Polytechnic Institute seminar on Mathematical Finance, the St. Gallen Finance Seminar, the KUL Finance Seminar and the Belgian Financial Research Forum. The empirical part of this work was produced while the authors were visiting the European Systemic Risk Board. We therefore acknowledge the outstanding support from the ESRB Secretariat and the research programme on derivatives led by the Joint Expert Group on Interconnectedness of the ESRB, chaired by Laurent Clerc and Alberto Giovannini. The authors declare no relevant or material financial interests that relate to the research described in this paper.}\\
{\normalsize  First Version: May, 2017 ~~~ This Version: April, 2019}
}
\date{}

\author{
  \textsc{Marco D'Errico} \vspace{-0.8cm}\\
  European Systemic Risk Board Secretariat,\vspace{-0.5cm}\\ European Central Bank
  \and
  \textsc{Tarik Roukny}\thanks{Corresponding author: tarik.roukny@kuleuven.be.\\ ~~Warmoesberg 26, 1000 Brussel - Belgium\\ Tel.: +32 16 32 65 11}\vspace{-0.8cm}\\
  Katholiek Universiteit Leuven\vspace{-0.5cm}\\
  Massachusetts Institute of Technology
}

\begin{document}

\newgeometry{left=20mm,right=20mm}

\maketitle
\vspace{-1.4cm}
\newpage
	\begin{abstract}
		\vspace{-0.3cm}
		\linespread{0.1}{
		Over-the-counter markets are at the center of the postcrisis global reform of the financial system. We show how the size and structure of such markets can undergo rapid and extensive changes when participants engage in portfolio compression, a post-trade netting technology. Tightly-knit and concentrated trading structures, as featured by many large over-the-counter markets, are especially susceptible to reductions of notional and reconfigurations of network structure resulting from compression activities. Using transaction-level data on credit-default-swaps markets, we estimate reduction levels consistent with the historical development observed in these markets since the Global Financial Crisis. Finally, we study the effect of a mandate to centrally clear over-the-counter markets. When participants engage in both central clearing and portfolio compression, we find large netting failures if clearinghouses proliferate. Allowing for compression across clearinghouses by-and-large offsets this adverse effect.}


	\end{abstract}

\restoregeometry

\clearpage

\setcounter{page}{1}
\restoregeometry

\section{Introduction}\label{sec:intro}

	Over-the-counter (OTC) markets held a central role during the Global Financial Crisis (GFC). As a result, several jurisdictions mandated major regulatory reforms including central clearing, increased capital requirements and enhanced trading transparency. Underpinning these initiatives was the need to contain counterparty risk emerging from excessive leverage and the lack of position transparency.\footnote{Cases such as the American International Group (AIG) have illustrated how the opacity of OTC markets generates a counterparty risk externality. \cite{acharya2014counterparty} show that this risk spills over from bilateral interdependencies and prevents the establishment of contracts with an adequately priced default risk premium. This externality, in turn, incentivises market participants to take on short positions with inefficiently large default risk. In general, it is too costly or infeasible in many realistic OTC market settings to fully internalize counterparty risks because it requires market participants to have the full information of the position of their counterparties.}


	A most notable effect of these reforms was the downsizing of several OTC markets. For example, the market for Credit-Default-Swaps (CDS) featured a remarkable reduction in size - from USD 61.2 trillion outstanding at end-2007 to USD 8.3 trillion outstanding at mid 2018.\footnote{See the Bank of International Settlement Statistics on OTC derivatives outstanding: \url{https://www.bis.org/statistics/derstats.htm}} In principle, this $86\%$ reduction could reflect a lowering in trading activity. However, several sources have instead attributed its origin to the massive adoption of a risk management technique prompted by the new regulation, namely portfolio compression.\footnote{Several reports from the Bank of International Settlement have provided evidence that the size reduction of OTC markets was limitedly associated with changes in transaction volumes. Instead, they have shown that it coincided with increased use of portfolio compression. See for example \cite{vause2010counterparty,schrimpf2015outstanding,aldasoro2018credit}. The private sector has also reported similar insights (see for example \cite{kaya2013reforming}). According to \cite{isda2015impact} - the International Swaps and Derivatives Association report - portfolio compression is responsible for a reduction of 67\% in total gross notional of Interest Rate Swaps. Media outlet such as The Economist and Financial Times have also addressed this issue at different points in time. See for example \cite{economist2008great,economist2009number,financialtimes2015compression,financialtimes2016otc}. Finally, TriOptima, a leader in the compression business, reports over one quadrillion USD in notional elimination through their service. Continuous updates are reported online \url{http://www.trioptima.com/services/triReduce.html}.}

	Portfolio compression is a netting mechanism which allows participants to coordinate contract replacements in order to reduce their portfolio of mutual obligations while maintaining the same underlying market risks. Figure \ref{fig-trilateral-compression} illustrates the process in a stylized example. Recognizing the need to limit excessive gross exposures,\footnote{According to \cite{cecchetti2009central}: ``Before the crisis, market participants and regulators focused on net risk exposures, which were judged to be comparatively modest. In contrast, less attention was given to the large size of their gross exposures. But the crisis has cast doubt on the apparent safety of firms that have small net exposures associated with large gross positions. As major market-makers suffered severe credit losses, their access to funding declined much faster than nearly anyone expected. As a result, it became increasingly difficult for them to fund market-making activities in OTC derivatives markets – and when that happened, it was the gross exposures that mattered.''} policy makers have broadly supported the adoption of portfolio compression (see Section~\ref{sec:background}). Importantly, other elements of the postcrisis regulatory environment have indirectly accelerated the private demand for portfolio compression. In fact, in addition to reducing counterparty risk, this technology allows participants to reduce both capital and margin requirements \citep{duffie2017financial}. Albeit the impact, analytical and empirical analysis have been so far limited due to the lack of adequate data and the opacity of the practice.

	In this paper, we present an analytical framework that explains how the size of OTC markets can be subject to large and rapid reductions when participants engage in portfolio compression cycles. We apply our framework to transaction-level data and estimate reduction ranges at par with the levels exhibited by the CDS markets. In addition, we study the interplay between central clearing - another major regulatory reform - and portfolio compression. We find large netting failures when clearinghouses proliferate and show that multilateral compression across clearinghouses can by-and-large compensate this adverse effect. 

	In our model, a market consists of a network of outstanding and fungible obligations among market participants. Netting opportunities exist when at least one participant intermediates the same obligation. The total amount of notional eligible for compression, henceforth \textit{market excess}, is further determined by the existence and length of chains of intermediation in the market. The exact fraction of excess that can be compressed is bounded by individual portfolio preferences and potential regulatory constraints. We study a spectrum of benchmark preference settings, resulting in so-called \textit{compression tolerances}, and investigate their feasibility and efficiency. These benchmarks differ in the extent to which participants will accept changes in their original sets of counterparty relationships. For instance, dealer banks may be indifferent vis-\`a-vis changes in their trades with other dealers while being conservative on the trading relationships they have established with their customers. We obtain a ranking of the full spectrum, highlighting a trade-off between the efficiency of a compression cycle and the degree of tolerances set by portfolio preferences: higher netting opportunities arise when participants are less conservative in their original sets of counterparties.

	Next, we estimate the levels of excess and compression efficiency empirically using a unique transaction-level dataset. To the best of our knowledge, this is the first calibrated analysis of the potential impact of a market-wide adoption of portfolio compression. We use a dataset consisting of all CDS contracts bought and sold by legal entities based in the European Union (EU) and all their counterparties. First, we find that the majority of markets exhibit levels of excess accounting for $75\%$ or more of their total notional size. Furthermore, even the most conservative scenario, in which all participants preserve their original trading relationships, eliminates on average more than $85\%$ of the excess in markets, for a total of at least two thirds of their original size.\footnote{These results and statistics are in line with evidence provided by several reports. See for example \cite{vause2010counterparty} for CDS globally and \cite{occ2016} for US derivatives.}

	These results are explained by the observed tightly-knit structure of the intra-dealer segment which allows for large excess elimination while preserving counterparty trading relationships. Nevertheless, we find that the efficiency of a conservative compression is impaired if market participants seek to bilaterally net out their positions beforehand. This effect is dampened when compression preferences are relaxed in the intra-dealer segment.

	\begin{figure}
			\centering
			\begin{subfigure}[t]{0.3\textwidth}
            
            	\begin{tikzpicture}
						\begin{scope}[every node/.style={circle,thick,draw}]
							\node (i) at (0,0) {i};
							\node (j) at (0,3) {j};
							\node (k) at (1.5,1.5) {k};
							\node (l) at (3.5,1.5) {l};
						\end{scope}
						\begin{scope}[every node/.style={fill=white,circle},
						    every edge/.style={draw=red,very thick}]
						    \path [->] (i) edge node {$5$} (j);
						    \path [->] (j) edge node {$10$} (k);
						    \path [->] (k) edge node {$20$} (i);
						    \path [->] (k) edge node {$10$} (l);
						\end{scope}
					\end{tikzpicture}
                \caption{Before Compression}
            \end{subfigure}
            \begin{subfigure}[t]{0.3\textwidth}
            	\begin{tikzpicture}
						\begin{scope}[every node/.style={circle,thick,draw}]
							\node (i) at (0,0) {i};
							\node (j) at (0,3) {j};
							\node (k) at (1.5,1.5) {k};
							\node (l) at (3.5,1.5) {l};
						\end{scope}
						\begin{scope}[every node/.style={fill=white,circle},
						    every edge/.style={draw=blue,very thick}]
						    \path [->] (j) edge node {$5$} (k);
						    \path [->] (k) edge node {$15$} (i);
					    \end{scope}
					    \begin{scope}[every node/.style={fill=white,circle},
						    every edge/.style={draw=red,very thick}]
						    \path [->] (k) edge node {$10$} (l);
						\end{scope}
					\end{tikzpicture}
                \caption{After Compression}
            \end{subfigure}
			\caption{A graphical example of portfolio compression. Panel (a) exhibits a market consisting of 4 market participants (i, j, k, l) with short and long positions on the same asset with different notional values. The aggregate gross notional of the market is 45. Panel (b) shows a possible compression solution to the market: by eliminating the obligations between i, j and k and generating two new obligations, the net position of each participant is unchanged while the gross positions of i, j and k have been reduced by 5. In aggregate, market size has been reduced by 15 units.}\label{fig-trilateral-compression}
	\end{figure}

	Finally, we run a stylized study of market excess and compression efficiency when participants adopt both portfolio compression and central clearing. Despite the multilateral netting opportunities brought by centralization, clearing also duplicates the notional value of each obligations. The effect of central clearing on market excess is therefore ambiguous, in particular, when multiple clearinghouses exist. When clearing takes place with one single central clearing counterparty (CCP), our empirically calibrated results indicate that this setting is dominated by - but close to - multilateral compression without central clearing. A proliferation of CCPs significantly and systematically increases the gap. Markets with several CCPs prevent large netting opportunities among common clearing members. Remarkably, we find that such effects are by-and-large reduced when a compression mechanism exist across CCPs, that is, when members of several clearinghouses can compress beyond their bilateral exposure to each clearinghouse independently. 

	The results of this paper contribute to several strands of the literature as well as ongoing policy debates. 

	Empirical studies including \cite{shachar2012exposing,benos2013structure,peltonen2014network,derrico2016flow,abad2016shedding,ali2016systemic} show that OTC markets are characterized by large concentration of notional within the intra-dealer segment. In particular, \cite{derrico2016flow} observe that in the global CDS market, intermediaries form a strongly connected structure which entails several closed chains of intermediation. The authors also show that between 70\% and 80\% of the notional in CDS markets is in the intra-dealer market across all reference entities. \cite{atkeson2013market} report that, in the US, on average, about 95\% of OTC derivatives gross notional is concentrated in the top five banks. \cite{abad2016shedding} report similar levels for interest rate swap, CDS and foreign exchange markets in the EU segment. This paper contributes by proposing a well defined measure of the market-level gap between gross and net notional. This so-called excess indicator, in turn, corresponds to the maximum amount of notional eligible for compression. Importantly, our results show that an explicit modeling of the entire network of bilateral obligation is necessary to estimate the efficiency of portfolio compression. Two distinct networks may exhibit the same level of excess while offering largely different opportunities for compression. We find that it is the combination of high notional concentration and long cycles of intermediation that allows for large reductions of excess even under conservative preferences.

	Theoretical analyses of OTC markets have addressed trading frictions and prices with a focus on the role of intermediaries \citep[see][]{duffie2005over,lagos2009liquidity,gofman2016efficiency,babus2017endogenous}. In particular, \cite{atkeson2015entry} and \cite{babus2013trading} find equilibrium conditions to observe large concentration in few market participants. On the one hand, \cite{babus2013trading} show the importance of the centrality of a dealer to reduce trading costs. On the other hand \cite{atkeson2015entry} show that only large participants can enter the market as dealers in order to make intermediation profits. More recently, studies have analyzed counterpary risk pricing in OTC markets. \cite{acharya2014counterparty} and \cite{frei2017managing} show that transparency and trade size limits usually prevent efficient risk pricing, respectively. While we study arbitrary sets of trading relationships, our results show that, under realistic assumptions, the adoption of post-trade technologies can largely impact the size of dealers - and the market as a whole - thus making such markets prone to rapid structural mutations when participants coordinate. This result is particularly revelant in lights of the role played by large and mispriced positions held by OTC dealers during the GFC as discussed by \cite{cecchetti2009central}. 

	The study of post-trade services has so far mainly focused on the costs and benefits of central clearing. \cite{duffie2011does} provide the ground work of this strand of works. The authors show that, while central clearing helps to reduce exposures at the asset class level, clearing heterogeneous asset classes removes the benefits of netting. \cite{cont2014central} explore the effect of heterogeneity across asset classes and show that a more risk sensitive approach to asset classes can alleviate the need to concentrate all netting activities in one single CCP. \cite{duffie2015central} and \cite{ghamami2017does} study the impact of clearing on collateral and capital requirements and show that trading costs can be higher or lower depending on the proliferation of CCPs and the extent to which netting opportunities can be exploited. The results of this paper on the effect of multiple CCPs provide a quantitative assessment of the loss in netting efficiencies and its impact on market excess. Furthermore, the finding that compression across CCPs vastly removes netting inefficiencies shows that multilateral compression among CCPs can address the trade-off between full centralization and efficiency losses introduced by \cite{duffie2011does}. 

	Regarding the theory of portfolio compression, \cite{o2014optimizing} stands as the main theoretical contribution. The author numerically analyzes the performances of different versions of compression algorithm on a synthetic network where all banks are connected. The author shows that, if performed optimally, compression mitigates counterparty risk. Our work differs in two main ways. First, we study sparse and concentrated market structures which correspond to a realistic setting distinguishing dealers from customers. In addition, we provide analytical solutions to the necessary and sufficient conditions for compression as a function of a spectrum of portfolio preferences. Finally, we apply our framework to transaction-level data and identify bounds of compression for each preference setting in OTC derivatives markets.

	Finally, our work relates to the growing stream of literature highlighting the important relationship between financial interconnectedness, stability and policy making \cite[see][]{allen2008networks,yellen2013interconnectedness}. These works explore the role of interdependencies on the propagation of distress \citep{allen2000financial,elliott2014financial,acemoglu2015systemic} and regulatory oversight \citep{alvarez2015mandatory,roukny2016interconnectedness,erol2017network,bernard2017bail}. Our paper shows how post-trade practices can affect the network of outstanding positions in financial markets. This matters both for the stability of such markets and for the tools required by policy makers to assess and address market stability. Compression reconfigures counterparty risk and intermediation chains which have held a central role in the propagation of distress during the 2007-2009 financial crises \citep{haldane2009rethinking,ecb2009counterpartycds}. 

	The rest of the paper is organized as follows. We provide an overview of the institutional background in Section \ref{sec:background}. In Section~\ref{sec:market}, we present our stylized model of OTC market and the  analysis of market excess. Section~\ref{sec:compression} presents a mathematical definition of portfolio compression; introduces benchmark preference settings; identifies feasibility and efficiency levels of each approach. In Section~\ref{sec:empiric} and~\ref{sec:compress-empirics}, we report the results of our empirical analysis of excess and compression efficiency in real OTC derivatives markets. In Section~\ref{sec:ccp}, we complement our framework with the addition of central clearing and study the impact on excess. Last, we conclude and discuss avenues for further research. The appendices provide proofs of the propositions and lemmas, additional results as well as the analytical details for the algorithms used in the paper. 

\section{Institutional Background\label{sec:background}}

	In contrast to centrally organized markets where quotes are available to all market participants and exchange rules are explicit, participants in OTC markets trade bilaterally and have to engage in search and bargaining processes. The decentralized nature of these markets makes them opaque as market information is often limited for most agents \citep{duffie2012dark}. In particular, the size, complexity and opacity of OTC derivatives markets have been a key target of the major regulatory reforms following the after-crisis meeting of the G-20 in September 2009. The summit resulted in a commitment to ``\textit{make sure our regulatory system for banks and other financial firms reins in the excesses that led to the crisis}"\footnote{Art. 16 of the Leader's Statement of the Pittsburgh Summit}. This initiative prompted two major financial regulatory reforms: the Dodd-Frank act in the US and the European Market Infrastructure Regulation (EMIR) in Europe. Such reforms include mandatory clearing of specific asset classes and standardized trading activity reports. In addition, the completion of the Basel III accords led to a general increases in capital and collateral requirements, especially regarding uncleared over-the-counter transactions.\footnote{Formally, the Markets in Financial Instrument Regulation (MiFIR) defines portfolio compression as follows: ``Portfolio compression is a risk reduction service in which two or more counterparties wholly or partially terminate some or all of the derivatives submitted by those counterparties for inclusion in the portfolio compression and replace the terminated derivatives with another derivatives whose combined notional value is less than the combined notional value of the terminated derivatives'' (see MiFIR, EU Regulation No 600/2014, Article 2 (47)). A similar definition is provided under the Dodd Franck act (see CFTC Regulation 23.500(h)).} This set of policy changes generated a large demand for novel services to accommodate the renewed regulatory environment \citep{fsb2017review}. In particular, efficient post-trade portfolio management became crucial to large financial institutions \citep{duffie2017financial}.

	Portfolio compression is a post-trade mechanism which exploits multilateral netting opportunities to reduce counterparty risk (i.e., gross exposures) while maintaining similar market risk (i.e., net exposures). The netting of financial agreements is a general process that can encompass different mechanisms. For example, \textit{close-out} netting is a bilateral operation that takes place after the default of one counterparty in order to settle payments on the net flow of obligations. In this respect, portfolio compression can be formally defined as a \textit{multilateral novation netting technique} that does not require the participation of a central clearinghouse. Rather than rejecting the participation of a central clearinghouse, this definition states that compression can be achieved even in the absence of a central counterparty. This distinction is relevant as \textit{multilateral netting} has often been equated only with central clearing. For sake of clarity and consistency with the current industry practices, we choose to articulate to remainder of the paper using the wordings related to \textit{compression}. 
	
	Over the last decade, the adoption of portfolio compression in derivatives markets has reportedly brought major changes. According to \cite{isda2015impact} - the International Swaps and Derivatives Association report - portfolio compression is responsible for a reduction of 67\% in total gross notional of Interest Rate Swaps. \cite{aldasoro2018credit} attributes the reduction of Credit Default Swap notional to a sixth of the levels exhibited a decade before to an extensive use of portfolio compression after the crisis. TriOptima, a leader in the compression business, reports over one quadrillion USD in notional elimination through their services.\footnote{Continuous updates are reported in \url{http://www.trioptima.com/services/triReduce.html}. Last check June 2017.}

	The mechanism of portfolio compression can also be seen as a multilateral deleveraging process operated without capital injection nor forced asset sales. Under the capital and collateral requirements resulting from the regulatory reforms, market participants engaging in portfolio compression are able to alleviate capital and collateral needs while preserving their capital structure and net market balances.\footnote{For instance, capital requirements under the Basel framework are computed including gross derivatives exposures \citep{basel2016capital}} Overall, we observe that the growing adoption of compression services has been driven by both incentives to improve risk management and adapt to the new regulatory requirements. 	

	In practice, multilateral netting opportunities can be identified only once portfolio information is obtained from several participants. However, it is individually undesirable for competing financial institutions to disclose such information among each other. Third-party service providers typically come at play to maintain privacy and provide guidance to optimize the outcome. To run a full compression cycle, compression services (i) collect data provided by their clients, (ii) reconstruct the web of obligations amongst them, (iii) identify optimal compression solutions and (iv) generate individual portfolio modification instructions to each client independently. 

	Portfolio compression has, in general, received a global regulatory support. For example, under the European Market Infrastructure Regulation (EMIR), institutions that trade more than 500 contracts with each other are required to seek to compress their trades at least twice a year.\footnote{See Article 14 of Commission Delegated Regulation (EU) No 149/2013 of 19 December 2012 supplementing Regulation (EU) No 648/2012 of the European Parliament and of the Council with regard to regulatory technical standards on indirect clearing arrangements, the clearing obligation, the public register, access to a trading venue, non-financial counterparties, and risk mitigation techniques for OTC derivatives contracts not cleared by a CCP (OJ L 52, 23.2.2013, p. 11- ``Commission Delegated Regulation on Clearing Thresholds'' or ``RTS'')} However, research on portfolio compression has been limited. In-depth analyses on the impact of portfolio compression for both markets micro-structure and financial stability has been lacking. 
\section{The model\label{sec:market}}

	We consider an Over-The-Counter (OTC) market composed of $n$ market participants denoted by the set $N = \{1,2,...,n\}$. These participants trade contracts with each other and establish a series of bilateral positions resulting in outstanding gross exposures stored in the $n\times n$ matrix $E$ with elements $e_{ij}\geq0$. The directionality departs from the seller $i$ to the buyer $j$ with $i,j \in N$. While we keep the contract type general, we assume that the resulting obligations are fungible: they have the same payoff structure from the market participants' perspective and can therefore be algebraically summed. The whole set of outstanding obligations in the market constitutes a financial \textit{network} or \textit{graph} $G$ = $(N, E)$. 




	The gross position of a market participant $i$ is the sum of all obligations' notional value involving her on either side (i.e., buyer and seller): $v_i^{\text{gross}} = \sum_j e_{ij} + \sum_j e_{ji} = \sum_j \left( e_{ij} + e_{ji}\right)$. The net position of market participant $i$ is the difference between the aggregated sides: $v_i^{\text{net}} = \sum_j e_{ij} - \sum_j e_{ji} = \sum_j \left( e_{ij} - e_{ji}\right)$. We also define the \textit{total gross notional} of the market as the sum of the notional amounts of all obligations: $x = \sum_{i}\sum_{j} e_{ij}.$

	Finally, market participants can either be customers or dealers.  Customers only enter the market to buy or sell a given contract. In contrast, dealers also intermediate between other market participants.\footnote{Dealers can also trade on their own. Hence, they do not necessarily have a matched position.} We use the following indicator to identify dealers in the market:

	\begin{equation}
      \delta(i) = \left\{
                  \begin{array}{lll}
                    1 ~~~~ \mbox{if} ~ \sum e_{ij}.\sum e_{ji} > 0 ~~~ & \mbox{\textbf{(dealer)}}\\
                    0 ~~~~ \mbox{otherwise} ~~~~~~ & \mbox{\textbf{(customer)}}
                  \end{array}
                \right. \nonumber
	\end{equation}
	    
	This framework generalizes the modeling approach of \cite{atkeson2015entry} with regard to market participant types. Three types of trading relationships can exist in the market: dealer-customer, dealer-dealer and customer-customer.
	\subsection{Market excess\label{sec:excess}}


	Figure \ref{fig-bowtie-illustration} shows the network of obligations of an actual OTC market for CDS contracts. Customers buying the CDS contract are on the left hand-side (green), customers selling the CDS contract are on the right hand-side and dealers are in the middle (blue and purple where purple nodes are the G-16 dealers). While buyers and sellers have a combined gross share of less than 5\%, their net position is equal to their gross position. In contrast, the set of dealers covers more than 95\% of gross market share while, on average, only one fifth is covered by net positions. As a result, 76\% of the notional held by dealers is the result of offsetting positions. Offsetting positions will constitute the netting set of interest for portfolio compression. 

	\begin{figure}
		\begin{center}
		
		\includegraphics[width=0.65\textwidth, trim = {2cm 1.5cm 1.5cm 4cm}]{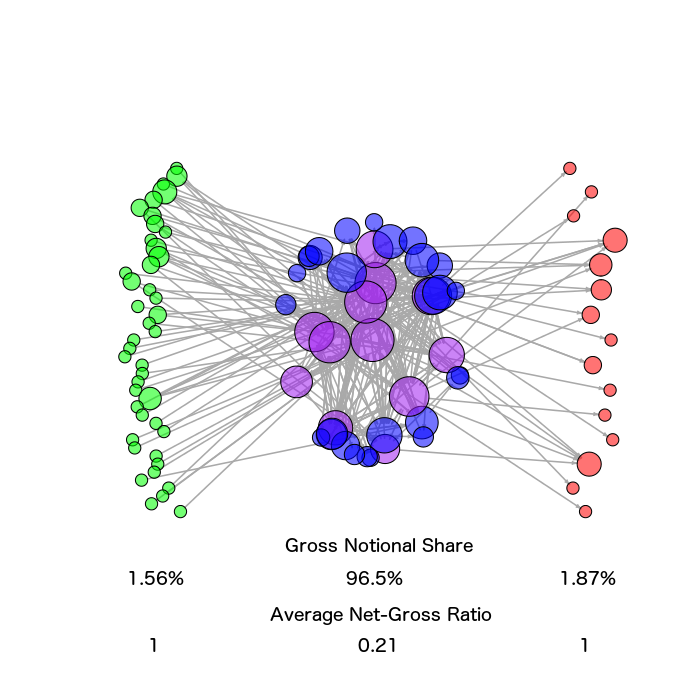}
		\end{center}
	    \caption{Network illustration of an OTC derivatives market, which maps all outstanding obligations for CDS contracts written on the same reference entity for the month of April 2016. The data were collected under the EMIR reporting framework and thus contain all trades where at least one counterparty is legally based in the EU. Green nodes correspond to buyers of the contract. Red nodes correspond to sellers of the contract. Purple nodes are G-16 dealers. Blue nodes are dealers not belonging to the G-16 dealer set. The first line below the network reports the share of gross notional based on individual positions for the segments: buyers, dealers, sellers. The second line reports the average net-gross ratio for each segment.
	    }
	    
	    \label{fig-bowtie-illustration}
	\end{figure}

	We now formalize the illustrated intuition to identify netting opportunities at the market level. Consider a post-trade operator $\Omega()$ that acts upon a market $G = (N,E)$ in order to modify the set of outstanding obligations: $G'=\Omega(G):(N,E)\rightarrow(N,E')$. Such operation is subject to different constraints. Here, we focus on the following \textit{net-equivalence} constraint which maintains the net position of each institution before and after the operation: $v^{net}_i = v^{\prime, net}_i,\forall i\in N$.


				
				


		Given an original market, it is possible to compute the minimum level of gross notional that can be obtained from a net-equivalent post-trade operation. This value corresponds to the net out-flow.\footnote{A concept similar to net out-flow has been partially adopted in other contexts, under the name of ``market open interest''. For instance, the Intercontinental Exchange clearinghouse defines open interest as ``the sum of the net notional for all participants that are net buyers of protection''. See \url{https://www.theice.com/marketdata/reports}}

		\begin{proposition}\label{prop:minimum-notional}
			Given a market $G=(N,E)$, if $\Omega$ is a net-equivalent operator on $G$, then:
		    \begin{eqnarray}
		    	G'& = & \Omega(G) =  \argmin_{x'=\sum e'_{ij}}(\Omega(G):(N,E)\rightarrow(N',E'))\Leftrightarrow \nonumber\\
		    	x^\prime &=& \frac{1}{2}\sum_{i=1}^n  |v_i^\text{net}| = \sum_{i:\; v_i^\text{net}> 0}^n  v^{net}_i = - \sum_{i: \; v_i^\text{net} < 0}^n  v^{net}_i .
		    			    \label{eq:minnotionalnetequivalent}
		    \end{eqnarray}
		\end{proposition}

		\begin{proof}
			~ Proof see Appendix~\ref{app-proofs}.
		\end{proof}


		Using Proposition \ref{prop:minimum-notional}, the total level of excess in a market will be the difference between the aggregate gross notional of a given market and the aggregate gross notional of the net-equivalent market with the minimum market aggregate gross notional. 

		\begin{definition}[Excess]\label{def:excess}
			The \textit{excess} in the market is defined as
			\begin{align}
				\Delta(G) & = x - x'\notag\\
				& = \left( \sum_{i=1}^n \sum_{j = 1}^n e_{ij} - \frac{1}{2}\sum_{i=1}^n  |v_i^{net}|\right) \label{eq:excess2}
			\end{align}
		\end{definition}

		The excess in the market is the amount of notional generated by obligations that offset each other. It corresponds to the maximum amount of notional that can be eliminated without affecting net positions.

		From Definition~\ref{def:excess}, we observe that the excess in a market is strictly positive if at least one participant exhibits a gross position larger than her net position. As we show below, such case only exists when the participant is a dealer.


		\begin{lemma}\label{lem:iff-excess}
			Given a market $G=(N,E)$,
		    \begin{equation}
				\sum_{i\in N} \delta(i)>0  \Leftrightarrow		   \Delta(G)>0.
		   		        \nonumber
		   \end{equation}
	   	\end{lemma}    
	    \begin{proof} 
		   ~ Proof see Appendix~\ref{app-proofs}.
		\end{proof}
		This result shows that intermediation decouples gross and net levels. Excess is a market-wide variable extending the individual level measurement introduced by \cite{atkeson2015entry}. Note that this result also explicitly shows why the existence of notional excess is intrinsic to OTC markets as stated in the following Corollary:\footnote{Note the special case of bilaterally netted positions. In the business practice of some instruments such as CDS contracts, two institutions sometimes terminate or reduce their outstanding bilateral position by creating an offsetting position (i.e., obligation of similar characteristics in the opposite direction). Such setting also generates excess. While this mechanism cannot be framed as intermediation per se, our formal network definition still applies. From a purely mathematical perspective, both participants are active on the buy and sell side and the related results remain.}

	    \begin{corollary}
	  	\label{cor:dealersexcess}  
	    	In the presence of dealers, over-the-counter markets always exhibit strictly positive notional excess.
	    \end{corollary}

		Even if some OTC markets exhibit customer-customer trading relationships, those interactions do not contribute to notional excess.\footnote{\cite{duffie2005over} stated the prevalent role of dealers in OTC markets. Furthermore \cite{atkeson2013market,abad2016shedding,derrico2016flow} have documented the high levels of notional concentration in the dealers segment of OTC markets, also illustrated in Figure~\ref{fig-bowtie-illustration}. We also document these features in the empirical Section of this paper.}
	\subsection{Market excess decomposition}

			Excess can be decomposed with respect to the intra-dealer segment and the customer segment, respectively. The intra-dealer (sub-)market only contains obligations between dealers while the customer (sub-)market contains obligations where at least one counterparty is a customer. Formally we have:

			\begin{definition}[Intra-dealer and customer market]\label{def:dealer-customer-markets}
			  	The set of obligations $E$ can be segmented in two subsets $E^D$ and $E^C$ such that
			  	\begin{equation}
			  		\delta(i).\delta(j) = 1 ~~~ \forall e_{ij} \in E^D \nonumber
				\end{equation}
				\begin{equation}
			  		\delta(i).\delta(j) = 0 ~~~ \forall e_{ij} \in E^C \nonumber
				\end{equation}
			 	Where $E^D$ is the intra-dealer market and $E^C$ is the customer market and $E^D + E^C = E$.
			\end{definition}

			We find that, in general, the excess is super-additive and cannot be decomposed.



		  	\begin{proposition}\label{prop:additivity-excess}
				Given a market $G=(N,E)$, and the two markets $G^1 = (N, E^1)$ and $G^2 = (N, E^2)$ obtained from the partition $\left\{E^1, E^2\right\}$ of $E$, then:
			    \begin{equation}
			    	\Delta(G) \geq \Delta(G^1) + \Delta(G^2) \nonumber
			    \end{equation}
		  	\end{proposition}
		  	\begin{proof}
		  		~ Proof see Appendix~\ref{app-proofs}.
		  	\end{proof}
		  	
			    


			This resul implies that $\Delta(N,E) \geq \Delta(N,E^D) + \Delta(N,E^C)$. We will use this result when considering different preference settings for different segments of the market.\footnote{Strict additiviy, $\Delta(N,E) = \Delta(N,E^D) + \Delta(N,E^C)$, exists when all dealers have a zero net position with regards to all their outstanding obligations with their dealer counterparties (i.e., $\sum^{dealer}_h(e_{dh} - e_{hd}) = 0, \forall d\in D$) or  with their customer counterparties ($\sum^{customer^+}_{c^+}e_{dc^+} - \sum^{customer^-}_{c^-}e_{c^-d} = 0, \forall d \in D$).}

\section{Portfolio Compression\label{sec:compression}}

	Building on the framework introduced in the previous section, we study how participants can coordinate to eliminate offsetting obligations using portfolio compression. For sake of simplicity, we do not explicitly model the incentives for participants to compress (see Section~\ref{sec:background}). Therefore, portfolio preferences are considered exogenous at this stage. For each sets of preferences, we identify when and how much excess can be eliminated. An analysis of endogenously driven equilibria is left for future research. 

	We define a compression operator as follows:

	\begin{definition}[Compression]\label{def:compression}
		Given a market $G =(N, E)$, a market $G^\prime = (N, E^\prime) := c(N,E)$ is \textit{compressed} w.r.t. to $G$ if and only if
		\begin{equation}
			v_i^{\prime net} = v_i^{net},  ~~ \mbox{for all} ~ i \in N ~~~~ \text{ and } x_i^{\prime} < x
			\nonumber
		\end{equation}
		where $c()$ is a net-equivalent network operator.
	\end{definition}

	Compression, at the market level, is thus an operation on the network of outstanding bilateral positions which reconfigures the set of obligations while (i) keeping all net positions constant (i.e., net-equivalence) and (ii) reducing gross notional at the market level.\footnote{At this stage, we assume the sets of participants to be the same before and after compression. However, some participants may be compressed out of the market if all their bilateral positions are eliminated. Furthermore, while we do not explicitly account for the role of a Central Clearinghouse so far, the framework accommodates the presence of such special market participant. We elaborate on this point in Section~\ref{sec:ccp}} 
	Admittedly, such definition is canonical and several refinements can be added to the compression operator. We discuss these aspects in Section~\ref{subsec:tolerances}.

		A direct corollary of Lemma \ref{lem:iff-excess} is that participants can effectively engage in portfolio compression only if the market exhibits intermediation.

	    \begin{corollary}\label{cor:condition-compression}
	    	Compression can only take place if there is intermediation in the market.
	    \end{corollary}

		The above result constitute a necessary but not sufficient condition. In addition to the net-equivalence condition, additional constraints can be set by participants' individual preferences and arbitrary regulatory policies to the compression exercise. The sufficiency condition must be expressed as a function of all applicable constraints which we express as tolerances.


	\subsection{Compression tolerances}\label{subsec:tolerances}

		The full design of a compression operator includes individual portfolio preferences and potential regulatory restrictions. For instance, market participants may be unwilling to compress some specific bilateral positions; policy-makers may prevent specific obligations from being created between some counterparties in the market. These multiple channels lead to several constraints on each potential bilateral pair of participants. As a result, for each possible obligation ($i,j$), we extract the most binding constraints. We refer to this selected set of constraints as \textit{compression tolerances}. Together, they limit the extent to which modifications can be brought to the original set of portfolios during the compression exercise. 
		Formally, compression tolerances form a set of bilateral limits in the following way

		\begin{definition}[Compression tolerances]\label{def:tolerances}
			A compression operator $c()$ s.t. $G' = (N,E') := c(N,E)$ is said to satisfy the set of compression tolerances $ \Gamma = \{(a_{ij},b_{ij}) | a_{ij},b_{ij} \in \mathbb{R}^{+}, i,j \in N\}$ if
		    \[
		    	a_{ij} \leq e'_{ij} \leq b_{ij} ~~ \forall (i,j) \in N^2
		    \]
		with $0 \leq a_{ij} \leq e_{ij} $, $e_{ij} \leq b_{ij}$ $\forall (i, j) \in N^2$.
		\end{definition}

		For each possible bilateral position in the resulting compressed market, there exist a lower ($a_{ij}$) and upper bound ($b_{ij}$). They determine the range of action accepted by the counterparty. Therefore a lower bound (resp. upper bound) cannot be higher (resp. lower) than the original obligation notional, i.e., $a_{ij} \leq e_{ij}$ (resp. $e_{ij} \leq b_{ij}$). Compression tolerances set the limits associated with each bilateral position and, consequently, determine how much excess can be eliminated.\footnote{In the context of compression service providers, compression tolerances determine how much the compression participants clients are willing not to alter their original positions. In derivatives markets, service providers such as TriOptima refer to these constraints as \textit{risk tolerances}. As they directly affect the efficiency of a compression exercise, bargaining can also take place between the service provider and its clients in order to modify those constraints. \textit{Dress rehearsals} are steps in the compression exercise where the service provider informs all the clients on a candidate compression solution and seeks their confirmation. Several iterations can be necessary before a solution satisfying all participants is reached.}.


		The set of all individual compression tolerances determines the exact set of offsetting obligations that can be included in the compression exercise. We distinguish between \textit{redundant} excess and \textit{residual} excess. The former is the excess that can be compressed while the latter is the excess that remains after compression. The determination of those levels is conditional upon (1) the underlying network of outstanding fungible obligations and (2) the set of all compression tolerances. Formally, we have:

		\begin{definition}[Residual and redundant excess]\label{def:residual-redundant-excess}
			A compression operator $c()$ s.t. $G' = (N, E') := c(N,E)$ satisfying the set of compression tolerances $ \Gamma = \{(a_{ij},b_{ij}) | a_{ij},b_{ij} \in \mathbb{R}^{+}, (i,j) \in N^2\}$ generates:
		    \begin{itemize}
		    	\item  $\Delta_{res}(G) = \Delta(G')$ ~~~ \textbf{(residual excess)}
		        \item $\Delta_{red}(G) = \Delta(G) - \Delta(G')  $ ~~~ \textbf{(redundant excess)}
		    \end{itemize}
		\end{definition}

		We have the following relationship: $\Delta(G) = \Delta_{res}(G) + \Delta_{red}(G)  $

	\subsection{Counterparty preference settings}

		In practice, compression tolerances are determined by a wide range of heterogeneous preferences from market participants and regulators. The space of possible compression tolerance combinations is theoretically infinite. In the following, we consider a general spectrum of preferences based on counterparty relationships. We start with two benchmark settings. In the first setting, participants are \textit{conservative}: they only allow for reductions of established obligations. In the second setting, participants are indifferent vis-\`a-vis changes in their trading relationships. These settings correspond to the following set of compression tolerances, $\{(a_{ij}, b_{ij}) = (0,e_{ij}) ~~ \forall (i,j) \in N^2\}$ and $\{(a_{ij}, b_{ij}) = (0,+\infty) ~~ \forall (i,j) \in N^2\}$, respectively:

		
		We refer to the first setting as \textit{conservative} and to the second as \textit{non-conservative}. Intuitively, the non-conservative case provides the highest levels of compression tolerances: it discards all counterparty constraints. The approach is deemed non-conservative with respect to the original web of obligations in the market. In the conservative case: compression tolerances are such that $e'_{ij} \leq e_{ij}$ for all bilateral positions. Hence, all participants are willing to reduce or eliminate their original obligation but no new relationship can be introduced between participants not trading ex-ante. It is conservative with respect to the original set of non-offsetting obligations in the market. Below, we formalize both settings.

			In the non-conservative compression setting, the resulting set of obligations $E'$ is not determined in any way by the previous configuration $E$.

			      \begin{definition} [Non-Conservative Compression]
			      		$c(N,E)$ is a non-conservative compression operator $\Leftrightarrow$ $c()$ is a compression operator that satisfies the compression tolerances set $\Gamma$:
			          \begin{equation}
			            a_{ij} = 0 ~~~ and ~~~ b_{ij} = +\infty, ~~ \forall (a_{ij},b_{ij}) \in \Gamma,  \nonumber
			          \end{equation}
			      \end{definition}

			In practice, such setting is unlikely to be the default modus operandi. However, it is conceptually useful to study as it sets the benchmark for the most compression tolerant case.

			In the conservative compression setting, the set of obligations in the compressed market is strictly obtained from reductions in notional values from the set of bilateral positions ex-ante. Formally, we have:

			    \begin{definition} [Conservative Compression]
			      $c(N,E)$ is a conservative compression operator $\Leftrightarrow$ $c()$ is a compression operator that satisfies the compression tolerances set $\Gamma$:
			      \begin{equation}
			      	a_{ij} = 0 ~~~ and ~~~ b_{ij} = e_{ij}, ~~ \forall (a_{ij},b_{ij}) \in \Gamma, (i,j) \in E \nonumber
			      \end{equation}
			      The resulting graph $G' = (N,E')$ is a `sub-network' of the original market $G = (N,E)$.
			    \end{definition}

			Such setting is arguably close to the way most compression cycles take place in derivatives markets.\footnote{We thank Per Sj\"{o}berg, founder and former CEO of TriOptima, for fruitful discussion on these particular points.} To illustrate the implementation of both approaches, we provide a simple example of a market consisting of 3 market participants in Appendix~\ref{app-example-3-market-participants}.
	\subsection{Compression feasibility and efficiency \label{sec:feasibility-efficiency}}


		In order to compare efficiencies under different tolerance sets, we associate each compression operator $c_k(N,E)$ with its relative redundant excess $\rho_k =\frac{\Delta^k_{red}(G)}{\Delta(G)}$. A compression operator over a market $G$, $c_s(N,E)$ is then more efficient than another compression operator, $c_t(N,E)$ if $\rho_s >\rho_t$.\footnote{Note that this efficiency ratio is invariant to scale transformations (see Appendix~\ref{app-rescaling-ratios} for details and derivations). This allows rescaling by an arbitrary amount without affecting the efficiency ratios (e.g., an exchange rate).}


		For each setting, we identify the conditions under which compression can take place and its related efficiency. We show the existence of a trade-off between the degree of portfolio conservation and the level of efficiency.

		\subsubsection{Non-conservative compression}

			Under non-conservative compression tolerances, original bilateral positions do not determine the outcome, only the net and gross positions of each participant do. We can thus generalize Corollary~\ref{cor:condition-compression} as follows: 
			\begin{proposition}\label{prop:non-conservative-feasibility}
				Given a market $G(N,E)$ and compression $c^{nc}()$ satisfying a non-conservative compression tolerance set $\Gamma$, 
			    \begin{equation}
			    	\Delta^{c{^{nc}}}_{\text{red}}(G)>0 \; \; \Leftrightarrow \; \;  \sum_{i\in N} \delta(i)>0\;. \nonumber
			    \end{equation}
			\end{proposition}	    
			\begin{proof}
			    ~ Proof see Appendix~\ref{app-proofs}.
		    \end{proof}

			Any compression exercise with a non-conservative set of tolerances is feasible if the market exhibits intermediation. In terms of efficiency, we obtain the following result:

			\begin{proposition}\label{prop:non-conservative-efficiency}
				Given a market $G = (N, E)$, there exists a set of non-conservative compression operators $C$ such that
				\begin{equation}
					C = \{c^{nc}|\Delta^{c^{nc}}_{res}(G) = 0\}\neq \emptyset \nonumber
				\end{equation}
				Moreover, let $G'=c^{nc}(G)| c^{nc} \in C$, then $G'$ is bi-partite.
			\end{proposition}
			\begin{proof}
				~ Proof see Appendix~\ref{app-proofs}.
			\end{proof}

			Non-conservative compression can always eliminate all the excess in a market. The proof of existence stems from the following generic algorithm: from the original market, compute all the net positions then empty the network of obligations and arbitrarily generate obligations such that the gross and net positions are equal. As net and gross positions are equal, the resulting market does not exhibit any intermediation. Recall from Lemma~\ref{lem:iff-excess} that if all intermediation chains are broken, the market exhibits no excess. We also obtain the following corollary:

			\begin{corollary}
				Given a market $G(N,E)$ and a compression operation $c(N,E)$,
			    \begin{equation}
			    	\Delta^c_{res}(G) = 0 \;\; \Leftrightarrow \;\; \sum_{i\in N'} \delta(i) = 0. \nonumber
			    \end{equation}
			\end{corollary}


			The resulting market is characterized by a bipartite underlying network structure.\footnote{A graph $G=(N,E)$ is bipartite if the set of nodes can be decomposed into two disjoint subsets $N^{out}$ and $N^{in}$ where each set is strictly composed of only one kind of node: respectively, nodes with only outgoing edges and nodes with only incoming edges. The edges are characterized as follows: $e_{ij}$ with $i \in N^{out}$ and $j \in N^{in}$. Also, a bipartite graph has no dealers: $\sum_{i\in N} \delta(i) = 0$} 

			For illustrative purposes, we provide a simple algorithm for this compression setting in Appendix~\ref{app-algorithms}.
			
		\subsubsection{Conservative compression}

			When compression tolerances are conservative, the compression operator can only reduce or eliminate existing offsetting obligations. In contrast with the non-conservative case, conservative compression cannot be applied to general chains of intermediation. Below we show that, only when chains of intermediation are closed, can conservative compression take place. 

			Let us first formalize the concept of closed intermediation chains:	

			\begin{definition}[Directed Closed Chain of Intermediation]
				A directed closed chain of intermediation is a set of obligations $K=\{e_{ij},e_{jt},..., e_{|K|i}\}$ in the graph $G=(N,E)$ such that $\prod_K e_{ij} > 0$.
			\end{definition}

			This structure constitutes the necessary and sufficient condition for conservative compression to be feasible:

			\begin{proposition}\label{prop:conservative-feasibility}
				Given a market $G(N,E)$ and a compression operator $c^c$ satisfying a conservative compression tolerance set $\Gamma$,
			    \begin{equation}
			    	\Delta^{c^c}_{red}(G)>0 \;\; \Leftrightarrow \;\; \exists E^* \subset E ~~~ \text {s.t. } \prod_{e^* \in E^*} e^* > 0. \nonumber
			    \end{equation}
			\end{proposition}
			\begin{proof}
				~ Proof see Appendix~\ref{app-proofs}.
			\end{proof}

			In contrast with the non-conservative approach, the efficiency of conservative compression is determined by the underlying network structure. In the following, we analyze the efficiency of conservative compression when applied to a dealer-customer network structure as is empirically observed in OTC markets.

			We start by showing that if the market exhibits a dealer-customer structure, conservative compression does not eliminate all the market excess.

			\begin{proposition}\label{prop:conservative-dealer-customer-efficiency}
				Given a market $G(N,E)$ and a compression operator $c()$ satisfying a conservative compression tolerance set $\Gamma$,
				\begin{equation}
					\exists ~i ~~s.t. 
					\begin{cases}
						\sum_j e^{C}_{ij} >0, ~~ e^{C}_{ij} \in E^{C}  \\
						\sum_j e^{C}_{ji} >0, ~~ e^{C}_{ji} \in E^{C} 
					\end{cases}
					\;\; \Rightarrow \;\; \Delta^{c{^c}}_{\text{res}}(G) > 0\nonumber
				\end{equation}
				Where $E^{C}$ is the set of dealer-customer obligations as defined under Definition \ref{def:dealer-customer-markets}.

			\end{proposition}
        	\begin{proof}
        		~ Proof see Appendix~\ref{app-proofs}.
    		\end{proof}

    		When dealers intermediate between customers on both sides (i.e., $\sum_j e^{C}_{ij}>0$ and $\sum_j e^{C}_{ji}>0$), the resulting chains of intermediation are necessarily open. In turn, they cannot be conservatively compressed and the residual excess of the compression is positive.

			In the case of a single closed chain of intermediation, the most efficient conservative compression procedure is given by the following result:

			\begin{lemma}\label{lemm:conservative-closed-chain}
				Given a directed closed chain $K=(N,E)$, consider the set of most efficient compression operations $C$ satisfying a conservative compression tolerance set $\Gamma$ then
			    \begin{equation}
				    e'_{ij} = e_{ij} - \min_{e}\{E\} ~~~ \forall e' \in E' ~~~~ \mbox{and} ~~~~ \Delta^c_{res}(K) = \Delta(K) - |E|\min_{e \in E} ~~~\forall c \in C\nonumber
			    \end{equation}
			\end{lemma}
			\begin{proof}
				~ Proof see Appendix~\ref{app-proofs}.
			\end{proof}

			Lemma \ref{lemm:conservative-closed-chain} shows that, in a single directed closed chain, eliminating the obligation with the lowest notional value and accordingly adjusting all other obligations in the chain is the most efficient conservative compression solution. The larger the length of the intermediation chain and the higher the minimum notional obligation value on the chain, the more excess can be eliminated conservatively. 	

			When the original market exhibits several closed chains of intermediation, the exact arrangement of chains in the network is critical to determine the resulting efficiency. In Appendix~\ref{app-conservative-compression}, we discuss cases of entangled chains (i.e., intermediation chains that share common obligations) with different ordering effects. In general, it is not possible to determine the residual excess of a conservative compression without further assumptions on the underlying structure. In order to guarantee a global solution, we characterize conservative compression as a linear programming problem and determine the most efficient compression procedure.\footnote{Details regarding the program characterization are provided in the Appendix~\ref{app-programming-charac}. }

			We can characterize the topological structure of the optimal solution. Let us define a Directed Acyclic Graph (DAG) as follow:

			\begin{definition}[Directed Acyclic Graph]
				A Directed Acyclic Graph (DAG) is a graph that does not contain any directed closed chains.
			\end{definition}

			We obtain the following result for any conservative compression solution:

			\begin{proposition}\label{prop:conservative-tree}
				Given a market $G(N,E)$ and a compression operator $c()$ satisfying a conservative compression tolerance set $\Gamma$.
				Let
				\[
					\mathbb{G} = \{G'|\Delta(G') = \min\{\Delta^{c^c}_{res}(G)\}
				\]
				then
				\[
					\forall G' \in \mathbb{G} ~~ s.t. ~~ \mbox{$G'$ is a DAG}
				\]
			\end{proposition}

			In fact, any closed chain of intermediation can be conservatively compressed. The above Proposition states that all optimal solutions will be characterized by an elimination of all closed chain resulting in an acyclic topological structure. Note that, as our objective function is set on the amount of excess that is removed, multiple directed acyclic solutions can, in principle, coexist.

			The results from Proposition~\ref{prop:conservative-feasibility}, Lemma~\ref{lemm:conservative-closed-chain} and Proposition~\ref{prop:conservative-tree} show that the set of closed chains of intermediation present in a market sets the efficiency of a conservative compression. More specifically, the number of closed chains, their length and their minimum notional obligation constitute the positive determinants of a tightly-knit market structure that partially generate larger efficiency gains for a conservative compression. The full determination requires knowledge on the exact market network structure. Section~\ref{sec:empiric} will empirical provide such analysis using transaction-level data.
	\subsection{Additional settings}
		\subsubsection{Hybrid compression}

		So far, we have focused on two benchmark preference settings. In more realistic settings, compression tolerances can be subject to the economic role of specific trading relationships. In the following, we consider a set of participants' preferences that combines properties from these two benchmarks. 

		\begin{assumption}
			Dealers prefer to keep their intermediation role with customers.
		\end{assumption}

		\begin{assumption}
			Dealers are indifferent vis-\`a-vis their bilateral positions with other dealers. Intra-dealer obligations can be switched at negligible cost.
		\end{assumption}

		The first assumption states that dealers value their role with customers. They will reject any compression solution that affects their bilateral positions with customers. Therefore, dealers set low compression tolerances on their customer related obligations.

		The second assumption posits that the intra-dealer network forms a club in which instances of a specific obligation do not signal a preference towards a given dealer counterparty. As a result, switching counterparties in the intra-dealer network has negligible costs in comparison with the overall benefits of compression. Therefore, dealers set high compression tolerances in the intra-dealer segment.

		Using Definition~\ref{def:dealer-customer-markets}, we have the following formal definition:

		\begin{definition} [Hybrid compression]
      		$c(N,E)$ is a hybrid compression operator i.f.f. $c()$ is a compression operator that satisfies the compression tolerances set $\Gamma$:
          	\begin{equation}
            	a_{ij} = 0 ~~~ and ~~~ b_{ij} = e_{ij}, ~~ \forall (a_{ij},b_{ij}) \in \Gamma, e_{ij} \in E^C \nonumber
          	\end{equation}
          	\begin{equation}
            	a_{ij} = 0 ~~~ and ~~~ b_{ij} = +\infty, ~~ \forall (a_{ij},b_{ij}) \in \Gamma, e_{ij} \in E^D \nonumber
          	\end{equation}

          \noindent			Where $E^C$ and $E^D$ are the customer market and the intra-dealer market, respectively, with $E^C+E^D = E$. 
      	\end{definition}

		The hybrid compression setting is a combination of (i) a non-conservative setting in the intra-dealer segment and (ii) a conservative setting in the customer segment.

		\begin{corollary}
			The feasibility conditions of the hybrid setting are
		    \begin{itemize}
		    \item non-conservative condition for $E^D$
		    \item conservative condition for $E^C$
		    \end{itemize}
		\end{corollary}

		Note that, in a dealer-customer market, a hybrid compression will only affect the intra-dealer segment because no closed chains of intermediation exist in the customer segment. As a result, the intra-dealer network will form a bipartite graph with zero residual intra-dealer excess.

		\begin{proposition}\label{prop:hybrid-efficiency}
			Given a market $G=(N,E)$, if 
		  	\begin{equation}
		  		\Delta(N,E) = \Delta(N,E^D) + \Delta(N,E^C) \nonumber
		  	\end{equation}
			\noindent
			then, a compression operator $c^h()$ satisfying a hybrid compression tolerance set $\Gamma$ leads to
			\begin{equation}
					\Delta^{c^h}_{res}(N,E) = \Delta(N,E^C) \nonumber
			\end{equation}
		\end{proposition}
		\begin{proof}
			~ Proof see Appendix~\ref{app-proofs}.
		\end{proof}

		In case the excess is additive, the efficiency of hybrid compression is straightforward. In case it is not, a specific algorithm must be implemented to obtain the exact level of efficiency (see Appendix~\ref{app-algorithms}).
		\subsubsection{Bilateral compression}\label{sec:bil-compression}

		Finally, we study a simple preference setting: bilateral compression. In this case, market participants do not exploit multilateral netting opportunities. Participants therefore do not need to share information and there is no need for a centralized mechanism. Formalizing this compression approach allows us, in part, to assess the added-value of a third party compression service provider when comparing efficiencies between bilateral and multilateral compressions. In our framework, bilateral compression is defined as follows:

		\begin{definition} [Bilateral compression]
      		$c(N,E)$ is a bilateral compression operator i.f.f. $c()$ is a compression operator that satisfies the compression tolerances set $\Gamma$:
          	\begin{equation}
            	a_{ij} = b_{ij} = \max{\{e_{ij}-e_{ji},0\}}, ~~ \forall (a_{ij},b_{ij}) \in \Gamma, e_{ij} \in E. \nonumber
          	\end{equation}
      	\end{definition}

      	For each pair of market participants $i$ and $j$, if we assume $e_{ij} > e_{ji}$, we have: $ e^{'}_{ij} = e_{ij} - e_{ji} $ and $e^{'}_{ji}= 0$ after bilateral compression.

      	In terms of feasibility, the mere existence of excess is not enough for bilateral compression to be applicable. In particular, we need at least two obligations between the same pair of counterparties and of opposite direction. Formally, we have the following results:

		\begin{proposition}\label{prop:bilateral-feasibility}
			Given a market $G(N,E)$ and a compression operator $c^b$ satisfying a bilateral compression tolerance set $\Gamma$,
		    \begin{equation}
		    	\Delta^{c^b}_{red}(G)>0 \;\; \Leftrightarrow \;\; \exists (i,j) \in N^2 ~~~ s.t.  ~~~ e_{ij}.e_{ji} > 0 ~~ \mbox{where} ~~ e_{ij},e_{ji}\in E \nonumber
		    \end{equation}
		\end{proposition}
		\begin{proof}
			~ Proof see Appendix~\ref{app-proofs}.
		\end{proof}

		The efficiency of bilateral compression is straightforward. It corresponds to the effect of netting out each pair of bilateral exposures. We thus obtain the following efficiency results:

		\begin{proposition}\label{prop:bilateral-efficiency}
			Given a market $G=(N,E)$ and a compression operator $c^b()$ satisfying a bilateral compression tolerance set $\Gamma$ leads to
			\begin{equation}
					\Delta^{c^b}_{res}(G) = \Delta(G) - \sum_{i,j \in N} \min \{e_{ij}, e_{ji}\} ~~ \mbox{where} ~~ e_{ij},e_{ji}\in E. \nonumber
			\end{equation}
		\end{proposition}
		\begin{proof}
			~ Proof see Appendix~\ref{app-proofs}.
		\end{proof}

		Technically, bilateral compression results in the removal of all closed chains of intermediation of length two. Hence, a bilaterally compressed market exhibit a maximum of one obligation between each pair of market participants.
	\subsection{Compression efficiency ranking\label{sec:comparing-benchmarks}}

		We close this Section with a ranking of efficiencies among the four benchmark settings we have introduced, namely, conservative, non-conservative, hybrid and bilateral. For each setting, we consider the maximum amount of excess that can be eliminated given the associated compression tolerances and the net-equivalent condition.

		\begin{proposition}\label{prop:efficiency-dominance}
			Given a market $G=(N,E)$ and the set of compression operators $\{c^c(),c^n(),c^h(),c^b()\}$ such that:
			\begin{itemize}
				\item $c^c()$ maximizes $\Delta^{c^c}_{red}(G)$ under a conservative compression tolerance set,
				\item $c^n()$ maximizes $\Delta^{c^n}_{red}(G)$ under a non-conservative compression tolerance set,
				\item $c^h()$ maximizes $\Delta^{c^h}_{red}(G)$ under a hybrid compression tolerance set,
				\item $c^b()$ maximizes $\Delta^{c^b}_{red}(G)$ under a bilateral compression tolerance set,
			\end{itemize}
			the following weak dominance holds:
			\begin{equation}
					\Delta^{c^b}_{red}(G) \leq \Delta^{c^c}_{red}(G) \leq\Delta^{c^h}_{red}(G) \leq\Delta^{c^n}_{red}(G) = \Delta(G)\nonumber
			\end{equation}
		\end{proposition}
		\begin{proof}
			~ Proof see Appendix~\ref{app-proofs}.
		\end{proof}

		This result shows a precise dominance sequence. First, we see that non-conservative compression is the most efficient. This stems from the fact that a global non-conservative solution always eliminates all the excess in a market (see Proposition~\ref{prop:non-conservative-efficiency}). The second most efficient compression operator is the hybrid compression, followed by the conservative. The least efficient approach is the bilateral compression. The loss in efficiency is due to the fact that bilateral compression cannot eliminate excess resulting from chains of length higher than two. The proof of this proposition derives from an analysis of the compression tolerance sets of each approach. In fact, it can be shown that the bilateral compression tolerance set is a subset of the conservative set which in turn is a subset of the hybrid set which is also a subset of the non-conservative set. This nested structure of compression tolerances ensures that any globally optimal solution of a superset is at least as efficient as the globally optimal solution of any subset. 

		Overall, this result shows a trade-off between efficiency in excess elimination and tolerances relative to changes in the underlying the web of outstanding obligations. The sequence from non-conservative compression to bilateral compression is a discrete gradient of relationship preservation. The more (resp. less) conservative, the less (resp. more) efficient. 

		Further analysis on the relative efficiencies of each approach (e.g., strong dominance, quantities, etc.) needs to include more detailed information on the underlying set of obligations $E$. Therefore, we proceed next with a empirical estimation based on transaction-level data.

\section{The data\label{sec:empiric}}

	\subsection{Outline}\label{subsec:analysis-strategy}
		
		In the following Sections, we apply our framework to transaction-level data of OTC markets. Assuming all market participants would engage in a portfolio compression cycle, we estimate the size reductions that such market would exhibit as a function of the sets of tolerances introduced in Section~\ref{sec:compression}.

		Determining the efficiency under each setting requires detailed knowledge of the bilateral obligations between counterparties. In general, such information for OTC markets is not readily available (see Section~\ref{sec:background}). Under EMIR, any legal entity based in the EU is required to report all derivatives trading activity to a trade repository. The European Systemic Risk Board (ESRB) is granted access to the collected data for financial stability purposes.\footnote{For more details on the dataset, the general cleaning procedure and other statistics, see \citep{abad2016shedding}}

		This unique dataset allows us to provide the first empirical account of the levels of market excess and the efficiency of various compression scenarios. In this paper, we focus on CDS derivatives.\footnote{Credit default swap contracts are the most used types of credit derivatives. A CDS offers protection to the buyer of the contract against the default of an underlying reference. The seller thus assumes a transfer of credit risk from the buyer. CDS contracts played an important role during the 2007-2009 financial crisis. For more information, see \citep{stulz2010credit}.} The dataset covers all CDS transactions and positions outstanding from October 2014 to April 2016 in which at least one counterparty is legally based in the European Union. 
		
		The reason we focus on the CDS market is fourfold. First, CDS contracts are a major instrument to transfer risk in the financial system. The key role they played in the unfolding of the GFC dramatically illustrates this point. Second, CDS markets have been early adopters of portfolio compression as discussed in Section~\ref{sec:background}. Third, the CDS markets we study are not subject to mandatory clearing and clearing rates remain low \citep{abad2016shedding}.\footnote{We focus on single-name CDS. In contrast with index CDS, there is no clearing mandate on those contracts under EMIR. See the Commission Delegated Regulation (EU) 2016/592 of 1 March 2016 supplementing Regulation (EU) No 648/2012 of the European Parliament and of the Council with regard to regulatory technical standards on the clearing obligation.} As such, they have maintained a dealer-customer structure relevant for non-trivial compression results. By the same token, they also lend themselves adequately to the central clearing counterfactual analyzing the introduction of mandatory clearing presented in Section~\ref{sec:ccp}. Fourth, the nature of these swaps make them the ideal candidate for our analysis. The notional amount of any bilateral contract corresponds to the expected payment (minus recovery rate) from the seller of protection to the buyer in case of default of the underlying entity. Therefore, positions are fungible as long as they are written on the same reference entity. In addition, it is always possible to identify, at any point in time, the payer and the receiver.\footnote{For other types of swaps, such as IRS, payer and receiver may change during the lifetime of a given trade and the overall analysis becomes less straightforward.}
		

		For each market, we compute the (i) dealer-customer network characteristics, (ii) excess statistics and (iii) efficiency under each tolerance setting: bilateral, conservative and hybrid compression.\footnote{We do not report results from non-conservative compression as an optimal solution always leads to zero residual excess by virtue of Proposition \ref{prop:non-conservative-efficiency}.} Bilateral compression is the result of a bilateral netting between all pairs of counterparties in the market as detailed in Section \ref{sec:bil-compression}. In the case of the conservative and hybrid compressions, we design a linear programming framework tailored to the respective tolerance sets.\footnote{All algorithms used to solve these problems are described in Appendix~\ref{app-algorithms}.} For each market $G$, we implement each compression algorithm and compute its efficiency: the ratio of redundant excess over the total level of excess:
		
		\begin{itemize}
			\item Bilateral: $\rho_{b} = \frac{\Delta^b_{red}(G)}{\Delta(G)}$;
		  	\item Conservative : $\rho_{c} = \frac{\Delta^c_{red}(G)}{\Delta(G)}$;
		  	\item Hybrid : $\rho_{h} = \frac{\Delta^h_{red}(G)}{\Delta(G)}$.
		\end{itemize}

		The resulting efficiency differences allows us to quantify i) the effect of coordinated multilateral compression (i.e., conservative and hybrid cases) versus asynchronous bilateral compression (i.e., bilateral case)\footnote{The synchronization aspect stems from the fact that both the conservative and hybrid approaches assume coordination among market participants. They all agree to compress the submitted observed trades at the same time. This condition is not necessary in the bilateral compression.} and ii) the quantitative effect of relaxing compression tolerances from bilateral to conservative to hybrid settings. In Appendix~\ref{app-bilateral-market-compression-analysis}, we report the same analysis for bilaterally compressed markets in order to quantify excess and compression efficiency beyond the bilateral redundancy. Results remain qualitatively robust.

		Finally, we compare results from applying multilateral compression on the original market and on the bilaterally compressed market. Doing so quantifies the potential losses in efficiency due to a sequence of bilateral-then-multilateral compression which bears policy design implications.

	\subsection{Dataset description}


		We use $19$ mid-month snapshots from October 2014 to April 2016. Overall, the original sample comprises 7,300 reference entities. The vast majority of the notional, however, is concentrated in a lower number of entities. We retain the top 100 reference entities which we find to be a good compromise between the amount of notional traded and clarity of analysis (see statistics in Section \ref{sec:gen-stat}).

		For each reference $k$, a market is the set of outstanding obligations written on $k$. Each bilateral position reports the identity of the two counterparties, the underlying reference entity, the maturity, the currency and its notional amount. We select the most traded reference identifier associated to the reference with the most traded maturity (by year) at each point in time. At the participant level, we select participants using their Legal Entity Identifier (LEI). In practice, financial groups may decide to submit positions coming from different legal entities of the same group. We do not consider such case in the remainder.\footnote{Our approach is in line with the recent Opinion on Portfolio Margining Requirements under Article 27 of EMIR Delegated Regulation of the European Securities and Market Authority (ESMA). Under articles 28, the netting sets related to different single name and indexes should be separated for portfolio margining comprise. Note that under article 29, different maturities can be considered the same product which is less conservative than in our approach.} 


		Our restricted sample comprises 43 sovereign entities (including the largest EU and G20 sovereign entities), 27 financials (including the largest banking groups) and 30 non-financials entities (including large industrial and manufacturing groups). We analyze each market separately.

	\subsection{Descriptive statistics}\label{sec:gen-stat}
		
		Table \ref{tab:bowtiestats-raw} provides the main statistics of each market segment.\footnote{Sampling statistics of the data are reported in the Appendix~\ref{app-sampling-stats}. The total notional of the selected 100 entities varies between 380Bn Euros and 480Bn Euros retaining roughly $30-34\%$ of the original total gross notional.} We compute the average number of dealers, customers on the buy side and customers on the sell side across all entities in the different snapshots.\footnote{Note that we empirically identify dealers as intermediaries beyond bilateral interactions. Indeed, from the formal definition of dealer in our framework (i.e., $\delta_i$), two market participants buying and selling from each other would be identified as dealers. This would not properly reflect the role of dealers in derivatives markets. As such, by convention, we set market participants as dealer if they appear as intermediary in when obligations are bilaterally netted out. Similarly, buying customers and selling customers are determined using the bilaterally compressed market. This convention does not affect the theoretical results and provides a more grounded interpretation of the empirical results, in particular for the hybrid compression.} We observe stable numbers across time: per reference entity, there are on average $18$ to $19$ dealers, $12$ to $17$ customers buying a CDS, $14$ to $21$ customers selling a CDS. The average number of bilateral positions per reference entity varies more through time but remains between $140$ and $170$. Taken as a whole, markets are quite sparse with an average density around $0.10$: $10\%$ of all possible bilateral positions between all market participants are realized. This measure is almost three times higher when we only consider the intra-dealer market. The bulk of the activity in those market revolves around intra-dealer trades. The amount of intra-dealer notional also highlights the level of activity concentration around dealers: it averages around $80\%$ of the total notional. These results are in line with the literature (see Section~\ref{sec:intro}). They provide evidence of the tightly-knit structure present in the intra-dealer segment. Finally, the last column of Table~\ref{tab:bowtiestats-raw} confirms the very low frequency of customer-customer trades: on average, less then $0.2\%$ of all obligations are written without a dealer on either side of the trade.

		\begin{sidewaystable}
			\centering
			\begin{tabular}{lccccccccc}

			Time & \begin{tabular}{c}Avg.\\ num.\\ dealers\end{tabular} & \begin{tabular}{c}Avg.\\ num. \\ customers \\buying\end{tabular} & \begin{tabular}{c}Avg.\\ num.\\ customers \\selling\end{tabular} & \begin{tabular}{c}Avg.\\ num.\\ obligations\end{tabular} &
			\begin{tabular}{c}Avg.\\ share\\ intra\\ dealer \\ notional \end{tabular} & \begin{tabular}{c}Avg. \\ density\end{tabular} & \begin{tabular}{c}Avg.\\ intra\\ dealer \\ density\end{tabular} & \begin{tabular}{c}Avg.\\ intra\\ customer \\ density\end{tabular} \\ 
			  \toprule
				Oct-14	&	18 &              16 &               20 &             153.72 &                           0.812 &           0.105 &                 0.332 &                    0.0010 \\
				Nov-14	&	18 &              16 &               21 &             162.26 &                           0.831 &           0.109 &                 0.345 &                    0.0006 \\
				Dec-14	&	19 &              17 &               21 &             171.13 &                           0.829 &           0.109 &                 0.339 &                    0.0005 \\
				Jan-15	&	19. &              17 &               21 &             171.25 &                           0.827 &           0.106 &                 0.334 &                    0.0006 \\
				Feb-15	&	19 &              17 &               21 &             168.45 &                           0.826 &           0.106 &                 0.335 &                    0.0004 \\
				Mar-15	&	18 &              15 &               17 &             154.42 &                           0.832 &           0.110 &                 0.339 &                    0.0007 \\
				Apr-15	&	18 &              13 &               15 &             143.71 &                           0.829 &           0.110 &                 0.344 &                    0.0005 \\
				May-15	&	18 &              12 &               15 &             143.66 &                           0.827 &           0.108 &                 0.336 &                    0.0008 \\
				Jun-15	&	18 &              13 &               14 &             142.01 &                           0.828 &           0.106 &                 0.323 &                    0.001 \\
				Jul-15	&	18 &              14 &               14 &             143.68 &                           0.813 &           0.101 &                 0.314 &                    0.0009 \\
				Aug-15	&	19 &              15 &               17 &             149.12 &                           0.821 &           0.101 &                 0.308 &                    0.0011 \\
				Sep-15	&	19 &              16 &               17 &             151.36 &                           0.804 &           0.098 &                 0.302 &                    0.0018 \\
				Oct-15	&	19 &              16 &               18 &             155.52 &                           0.815 &           0.099 &                 0.297 &                    0.0013 \\
				Nov-15	&	19 &              17 &               19 &             158.10 &                           0.810 &           0.099 &                 0.293 &                    0.0017 \\
				Dec-15	&	19 &              17 &               19 &             158.78 &                           0.821 &           0.098 &                 0.292 &                    0.0012 \\
				Jan-16	&	20 &              17 &               18 &             159.58 &                           0.822 &           0.098 &                 0.291 &                    0.0013 \\
				Feb-16	&	19 &              17 &               18 &             159.80 &                           0.813 &           0.098 &                 0.291 &                    0.0012 \\
				Mar-16	&	18 &              14 &               17 &             144.11 &                           0.790 &           0.096 &                 0.301 &                    0.0018 \\
				Apr-16	&	19 &              14 &               17 &             146.42 &                           0.811 &           0.098 &                 0.301 &                    0.0019 \\       
			   \bottomrule
			\end{tabular}
			\caption{Statistics of sampled markets over time: average numbers of dealers, customers, obligation and concentration statistics.}\label{tab:bowtiestats-raw}
		\end{sidewaystable}

\section{Market excess and compression efficiency}\label{sec:compress-empirics}

	We start by measuring the level of excess present in the original markets as a function of the total gross notional (i.e., $\epsilon(G)  = \frac{\Delta(G)}{x}$). Table~\ref{tab:excess-stats-raw} reports the statistics of excess computed across all reference entities for six snapshots equally spread between October $2014$ and April $16$ including minimum, maximum, mean, standard deviation and quartiles, computed across all 100 reference entities in our sample. Results on the means and medians are stable over time and mostly higher than $0.75$. The interpretation of this result is that around three quarters of the gross notional in the most traded CDS markets by EU institutions is in excess vis-\'a-vis participants' net position. At the extremes, we note a high degree of variability: the minimum and maximum levels of excess relative to total gross notional oscillate around $45\%$ and $90\%$ respectively. 

	\begin{table}
		\centering
		\begin{tabular}{lrrrrrrr}
			\textbf{Total Excess}  & Oct-14 & Jan-15 & Apr-15 & Jul-15 & Oct-15 & Jan-16 & Apr-16 \\ 
			\toprule
			min          &  0.529 &  0.513 &  0.475 &  0.420 &  0.533 &  0.403 &  0.532 \\
			max          &  0.904 &  0.914 &  0.895 &  0.901 &  0.903 &  0.890 &  0.869 \\
			mean         &  0.769 &  0.777 &  0.766 &  0.757 &  0.751 &  0.728 &  0.734 \\
			stdev        &  0.077 &  0.082 &  0.085 &  0.090 &  0.082 &  0.096 &  0.080 \\
			first quart. &  0.719 &  0.733 &  0.712 &  0.703 &  0.693 &  0.660 &  0.678 \\
			median       &  0.781 &  0.791 &  0.783 &  0.769 &  0.758 &  0.741 &  0.749 \\
			third quart. &  0.826 &  0.847 &  0.832 &  0.822 &  0.808 &  0.802 &  0.796 \\
		\end{tabular}
		\caption{Statistics of market excess over time: share of notional in excess against total gross notional for each market.}
		\label{tab:excess-stats-raw}
	\end{table}
	
	Overall, results reported in Table~\ref{tab:excess-stats-raw} show that large amounts of notional are eligible for compression. We now move to the efficiency of each compression operator. The results are reported in Table~\ref{tab:compression-results-raw}. After having implemented the compression algorithms on each market, we compute efficiency as defined in Section \ref{subsec:analysis-strategy}.\footnote{Note the current compression exercise does not represent the amount of compression achieved in the market. Rather, this exercise identifies levels still achievable given the empirically observed outstanding positions.}

	Analyzing the means and medians, we observe that the bilateral compression already removes $50\%$ of excess on average. Nevertheless both multilateral compression approaches (i.e., conservative and hybrid) outperform it by removing around $85\%$ and $90\%$ of the excess respectively. Levels are larger than the maximum efficiency achievable by bilateral compression which oscillates around $75\%$. In comparison with the bilateral efficiency, the conservative and hybrid approaches perform similarly on the extremes: minima range between $55\%$ and $62\%$ and maxima range between $98\%$ and $99\%$, respectively. In particular, results from the conservative compression show that, even under severe constraints, the vast majority of market's excess can be eliminated. This result is made possible by the large levels of concentrations and tightly-knit structure exhibited in the intra-dealer segment.

	\begin{table}
		\centering
		\begin{tabular}{lrrrrrrr}
			\textbf{Bilateral ($\rho_b$)}  & Oct-14 & Jan-15 & Apr-15 & Jul-15 & Oct-15 & Jan-16 & Apr-16 \\ 
			\toprule
			min          &  0.278 &  0.281 &  0.286 &  0.277 &  0.276 &  0.276 &  0.260 \\
			max          &  0.779 &  0.791 &  0.759 &  0.777 &  0.717 &  0.711 &  0.746 \\
			mean         &  0.528 &  0.536 &  0.524 &  0.522 &  0.513 &  0.512 &  0.543 \\
			stdev        &  0.101 &  0.106 &  0.103 &  0.105 &  0.107 &  0.109 &  0.108 \\
			first quart. &  0.464 &  0.460 &  0.469 &  0.452 &  0.448 &  0.444 &  0.448 \\
			median       &  0.526 &  0.542 &  0.535 &  0.530 &  0.517 &  0.528 &  0.555 \\
			third quart. &  0.583 &  0.597 &  0.590 &  0.600 &  0.596 &  0.597 &  0.623 \\
			\\
			\textbf{Conservative ($\rho_c$)}  & Oct-14 & Jan-15 & Apr-15 & Jul-15 & Oct-15 & Jan-16 & Apr-16 \\ 
			\toprule
			min          &  0.558 &  0.547 &  0.545 &  0.507 &  0.491 &  0.528 &  0.574 \\
			max          &  0.985 &  0.982 &  0.973 &  0.967 &  0.968 &  0.979 &  0.969 \\
			mean         &  0.836 &  0.857 &  0.848 &  0.843 &  0.828 &  0.827 &  0.834 \\
			stdev        &  0.091 &  0.087 &  0.090 &  0.091 &  0.104 &  0.106 &  0.090 \\
			first quart. &  0.781 &  0.816 &  0.810 &  0.800 &  0.777 &  0.773 &  0.788 \\
			median       &  0.852 &  0.880 &  0.868 &  0.858 &  0.849 &  0.847 &  0.860 \\
			third quart. &  0.906 &  0.925 &  0.913 &  0.915 &  0.902 &  0.907 &  0.904 \\
			\\
			\textbf{Hybrid ($\rho_h$)}  & Oct-14 & Jan-15 & Apr-15 & Jul-15 & Oct-15 & Jan-16 & Apr-16 \\ 
			\toprule
			min          &  0.589 &  0.626 &  0.636 &  0.653 &  0.574 &  0.619 &  0.676 \\
			max          &  0.990 &  0.994 &  0.988 &  0.990 &  0.994 &  0.989 &  0.990 \\
			mean         &  0.878 &  0.898 &  0.894 &  0.893 &  0.881 &  0.882 &  0.898 \\
			stdev        &  0.079 &  0.072 &  0.074 &  0.073 &  0.085 &  0.080 &  0.069 \\
			first quart. &  0.821 &  0.859 &  0.862 &  0.865 &  0.831 &  0.836 &  0.863 \\
			median       &  0.894 &  0.916 &  0.918 &  0.912 &  0.901 &  0.908 &  0.911 \\
			third quart. &  0.935 &  0.952 &  0.947 &  0.951 &  0.948 &  0.945 &  0.947 \\
		\end{tabular}
		\caption{Statistics of compression efficiency over time: share excess eliminated after compression against original level of market excess for each market.
		}
		\label{tab:compression-results-raw}
	\end{table}

	Analyzing further the interplay between bilateral and multilateral compression showcases the added-value of multilateral compression services. In fact, participants can engage in a decentralized and asynchronous fashion to achieve bilateral compression. This is not straightforward for multilateral compression. This difference also allows participants to seek to bilaterally compress some of their positions before participating in a multilateral compression cycle. We analyze this situation as follows: for each setting, we compare the efficiency of the operation on the original market with the aggregate efficiency when bilateral compression is applied first.\footnote{Appendix~\ref{app-bilateral-market-compression-analysis} reports an analysis on bilaterally compressed markets similar to the one produced for original markets.} 

	Figure~\ref{fig:raw-vs-bilateral-compression-comparison} reports the distribution of efficiency ratios when multilateral compression operators are applied to the full network and when they are combined with bilateral compression first. The latter results are obtained by adding the absolute bilateral results reported in Table~\ref{tab:compression-results-raw} to the absolute excess reduction for the conservative and hybrid approach as in Table~\ref{tab:compression-results-bil} then dividing by the aggregate notional of the original markets.

	\begin{figure}[!tbp]
	\centering
		\includegraphics[width = \textwidth]{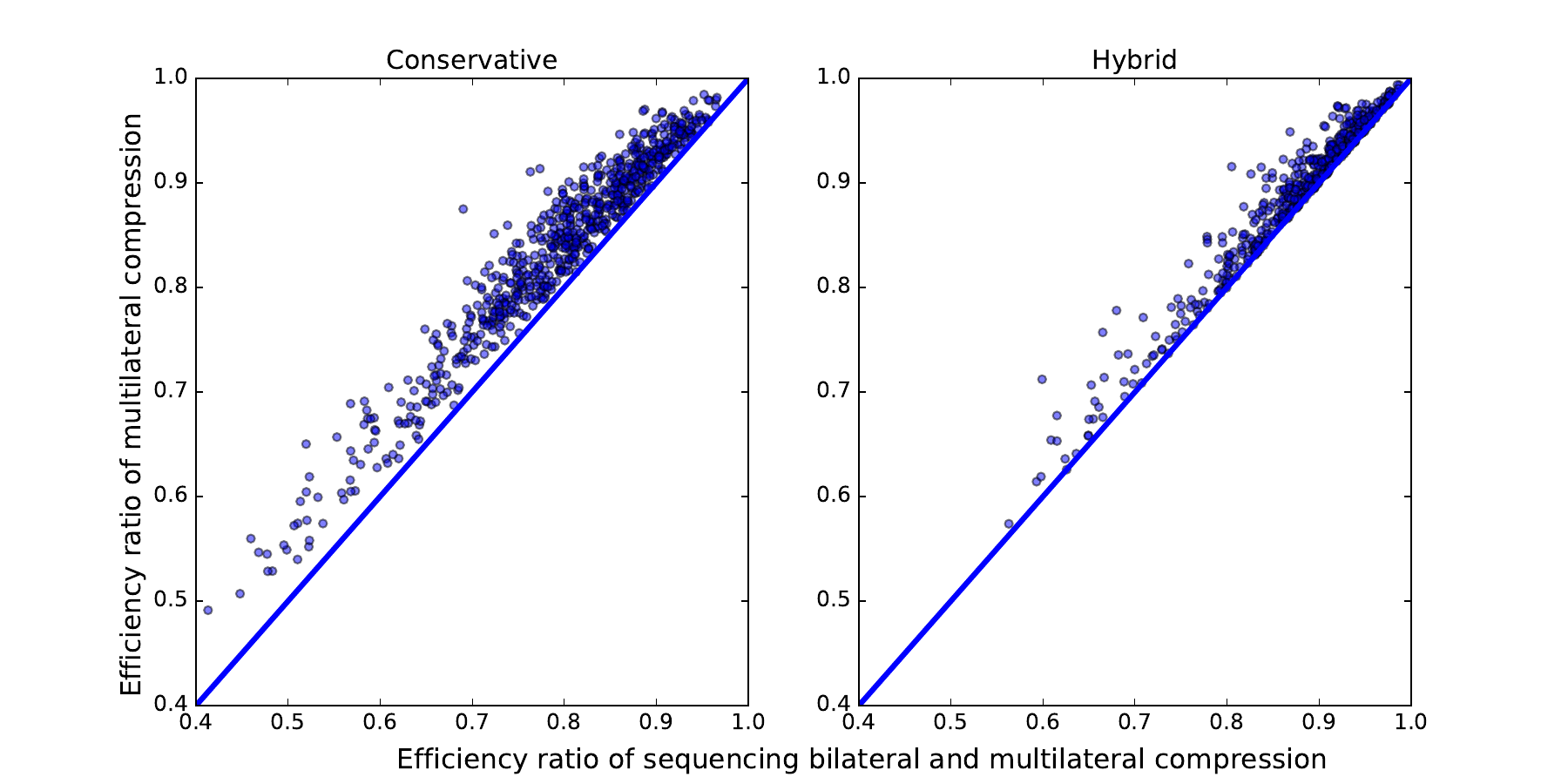}
		\caption{Comparison of the efficiency between multilateral compression in the original markets and a sequence of bilateral and multilateral compression. All snapshots and market instances are reported on the same figures.\label{fig:raw-vs-bilateral-compression-comparison}}
	\end{figure}

	The results show that multilateral compression on the original market is always more efficient than the sequence of bilateral-then-multilateral compression. Nevertheless, the sequence is particularly relevant under a conservative setting of preferences. In fact, the difference for hybrid compression is lower (i.e., about one percentage point improvement in the median) than in the conservative case (i.e., up to seven percentage points). 

	More in general, Figure~\ref{fig:raw-vs-bilateral-compression-comparison} suggests that a more coordinated and collective action for compression provides more efficiency. Henceforth, regulatory incentives would be more effective when favoring multilateral over bilateral compression. However, under EMIR, while there is no explicit distinction, the condition is set at the bilateral level (i.e., 500 bilateral contracts with the same counterparty), which may encourage bilateral compression. In contrast, measures based on notional approaches such as net-to-gross ratios would potentially improve incentives to compress as well as the efficiency of the multilateral exercises.
\section{The effects of clearinghouse proliferation\label{sec:ccp}}

	The promotion of central clearing in OTC derivatives markets has been a major element of the post-crisis regulatory reform \citep{fsb2017review}. Central clearing consists of interposing a Central Clearing Counterparty (CCP) between each side of a contract. The guiding principle of the reform is based on the premise that increased clearing of transactions provides more stability to markets by means of counterparty risk elimination, increased netting efficiencies and risk mutualization \citep{cecchetti2009central}. Mandates to clear specific asset classes together with larger capital and collateral costs for non-cleared transactions have deeply transformed the organization of several OTC markets in the recent years.\footnote{For an overview of the post-crisis clearing incentives and their effects on OTC markets see \citep{fsb2018review}} 

	In this section, we adapt our framework to assess the effect of several central clearing scenarios on market excess. In particular, we are interested in estimating compression performances in the presence of central clearing. 

	Note that there exists a distinction between central clearing and portfolio compression in terms of netting. In practice, CCPs generate multilateral netting opportunities among their members in the form of cash flows. This does not necessarily translate into a reduction of gross positions.\footnote{For instance, the London Clearing House (LCH) has developed a platform to allow their clearing members to access compression service providers such as TriOptima over the trades cleared by LCH. See: \url{https://www.lch.com/services/swapclear/enhancements/compression}} Portfolio compression instead explicitly involves the termination of positions, implying a systematic reduction of gross notional. For sake of consistency, in what follows we will consider the netting efficiency of central clearing in terms of gross reduction when trades are bilaterally compressed with the CCP.

	\subsection{A single CCP}

		Introducing a single CCP transforms the network structure of a market into a star network where the CCP, denoted $c$, is on one side of all obligations. Every original trade is novated into two new trades. By construction, the CCP has a net position of 0 and its gross position is equal the total market size: $v^{gross}_{c} = \sum_{c,j}e^{'}_{cj} = x^{c}$. Before the bilateral compression with the CCP, we have
		$x^{CCP} = 2x ~~~ \mbox{and} ~~~ v^{c}_i = v_i ~ \forall i \in N$. In fact, the total size of the market doubles with a CCP while all market participants keep their net position unchanged. Let $m$ be the minimum total notional required to satisfy every participants' net position as defined in Eq~\ref{eq:minnotionalnetequivalent} from Proposition~\ref{prop:minimum-notional}. Hence, the excess before compression is given by: $\Delta(G^{CCP}) = 2x - m = x + \Delta(G)$.

		Compression in a single CCP market is equivalent to the bilateral compression of a star-network. All trades between a counterparty $i$ and the CCP $c$ are bilaterally netted such that: $e^{\prime}_{ic} = max \{e_{ic}-e_{ci},0\} ~~ \mbox{and} ~~ e^{\prime}_{ci} = max \{e_{ci}-e_{ic},0\}$. As a result, the total size of the market after compression with a single CCP is given by $x' = 2m$.
		
		Compressing the original market $G=(N,E)$ under one single CCP thus leads to a redundant excess of $\Delta^{CCP}_{red}(G) = x - x^{\prime} = x - 2m$. We can thus compute the efficiency ratio as follows $\rho^{CCP} = \frac{\Delta^{CCP}_{red}(G)}{\Delta(G)} = \frac{x - 2m}{x - m} = 1 - \frac{m}{x-m}$.

		Without loss of generality, we formulate the efficiency under one single CCP as follows:
		\begin{proposition}
			\begin{equation}\label{eq:1-ccp-efficiency}
				\rho^{CCP} = 1 - \frac{m + x - x}{x-m}= 2 - \frac{x}{x-m} = 2 - \frac{1}{\epsilon}
			\end{equation}		
		where $\epsilon$ is the share of excess present in the original market: $\epsilon = \frac{\Delta(G)}{x} = \frac{x-m}{x}$.
		\end{proposition}
		From this expression we see that:
		\begin{corollary}\label{cor:1-ccp}
			\begin{itemize}
				\item If the excess in a market is less than $50\%$ of total notional, compressing with a single CCP is \textbf{counter-efficient}: it increases the excess
				\item If the excess in a market is equal to $50\%$ of total notional, compressing with a single CCP is \textbf{neutral}: it does not modify the excess
				\item If the excess in a market is higher to $50\%$ of total notional, compressing with a single CCP is \textbf{sub-excess ratio efficient}: the efficiency is always lower than the excess share.
				\item If the excess in a market is equal to $100\%$ of total notional, compressing with a single CCP is \textbf{fully efficient}: it removes all the excess.
			\end{itemize}
		\end{corollary}

		From the outcomes presented in Corollary \ref{cor:1-ccp}, we identify the most empirically relevant case using the data described in Section~\ref{sec:empiric}. We compare the efficiency of one single CCP with the efficiency results from the previous Section. For each compression setting, we collect the full set of markets - through references and time - and compare the efficiency ratios with Equation~\ref{eq:1-ccp-efficiency}.  

		Figure \ref{fig:ccp-compression} reports the results. The multilateral compression operations (conservative and hybrid) systematically yield higher efficiency than the compression with one single CCP. In the majority of cases, bilateral compression is less efficient than one single CCP.

		\begin{figure}
		\centering
			\includegraphics[width = \textwidth]{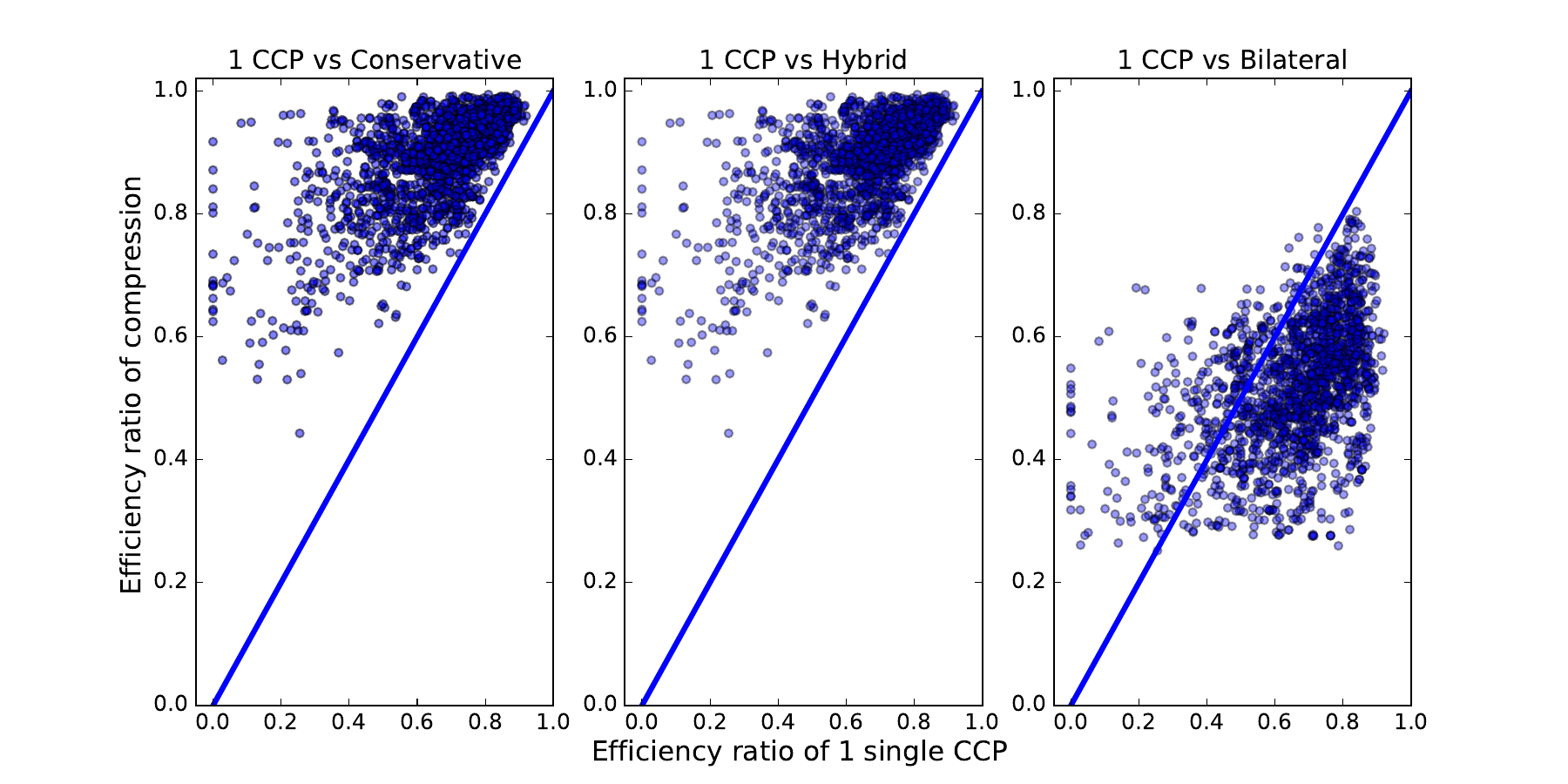}
			\caption{Comparison of efficiency ratios between compression operators and one single CCP. All snapshots and market instances are reported on the same figures.\label{fig:ccp-compression}}
		\end{figure}

		Despite the multilateral netting opportunities brought by centralization, novating contracts to the clearinghouse also duplicates the notional value of each bilateral obligations. When only considering the effects over gross notional, the above empirical exercise indicates that this trade-off is first-order dominated by multilateral compression without central clearing.

	\subsection{Multiple CCPs}

		We consider the case of multiple CCPs in the market. We run an empirical exercise in which the set of bilateral positions is reorganized among several CCPs. For a number $n^{ccp}$ of CCPs, each bilateral position is cleared with one CCP chosen uniformly at random. Once all bilateral positions are assigned and duplicated, each CCP compresses bilaterally with their members. For each given market and $n^{ccp}$, we generate 1000 realizations of CCP allocations and compute statistics of the compression efficiency ratios.

		We study two cases. In the first scenario, only bilateral compression between members and their CCPs can take place. In the second scenario, we analyze the effect of multilateral compression across CCPs. Compression across CCPs can take place when members of one CCP are also members of another CCP. In the following, we assume a conservative preference setting among CCPs. Figure~\ref{fig-cross-ccp-compression} provides a stylized example of compression across two CCPs. Note that compression tolerances on counterparty relationships among participants become irrelevant in this context: all participants are exclusively exposed to CCPs. 

		\begin{figure}
			\centering
			\begin{subfigure}[t]{0.5\textwidth}
            
            	\begin{tikzpicture}
						\begin{scope}[every node/.style={circle,thick,draw}]
							\node (i) at (0,0) {i};
							\node (j) at (0,3) {j};
							\node (m1) at (3,3) {...};
							\node (m2) at (3,0) {...};
							\node (m3) at (-3,3) {...};
							\node (m4) at (-3,0) {...};
						\end{scope}
						\begin{scope}[every node/.style={rectangle,thick,draw}]
							\node (k) at (1.5,1.5) {2};
							\node (l) at (-1.5,1.5) {1};						
						\end{scope}

						\begin{scope}[every node/.style={fill=white,circle},
						    every edge/.style={draw=red,very thick}]
						    \path [->] (l) edge node {$30$} (j);
						    \path [->] (j) edge node {$15$} (k);
						    \path [->] (k) edge node {$20$} (i);
						    \path [->] (i) edge node {$10$} (l);
						\end{scope}
						\begin{scope}[every node/.style={fill=white,circle},
						    every edge/.style={draw=gray,very thick}]
						    \path [->] (k) edge node {$15$} (m1);
						    \path [->] (m2) edge node {$20$} (k);
						    \path [->] (l) edge node {$10$} (m3);
						    \path [->] (m4) edge node {$30$} (l);
					    \end{scope}
					\end{tikzpicture}
                \caption{Before Compression}
            \end{subfigure}
            \begin{subfigure}[t]{0.3\textwidth}
            	\begin{tikzpicture}

						\begin{scope}[every node/.style={circle,thick,draw}]
							\node (i) at (0,0) {i};
							\node (j) at (0,3) {j};
							\node (m1) at (3,3) {...};
							\node (m2) at (3,0) {...};
							\node (m3) at (-3,3) {...};
							\node (m4) at (-3,0) {...};
						\end{scope}
						\begin{scope}[every node/.style={rectangle,thick,draw}]
							\node (k) at (1.5,1.5) {2};
							\node (l) at (-1.5,1.5) {1};						
						\end{scope}

						\begin{scope}[every node/.style={fill=white,circle},
						    every edge/.style={draw=red,very thick}]
						\end{scope}

						\begin{scope}[every node/.style={fill=white,circle},
						    every edge/.style={draw=gray,very thick}]
						    \path [->] (k) edge node {$15$} (m1);
						    \path [->] (m2) edge node {$20$} (k);
						    \path [->] (l) edge node {$10$} (m3);
						    \path [->] (m4) edge node {$30$} (l);
					    \end{scope}

						\begin{scope}[every node/.style={fill=white,circle},
						    every edge/.style={draw=blue,very thick}]
						    \path [->] (l) edge node {$20$} (j);
						    \path [->] (k) edge node {$10$} (i);
						    \path [->] (j) edge node {$5$} (k);
						    \path [->] (i) edge node {$0$} (l);
					    \end{scope}

					\end{tikzpicture}
                \caption{After Compression}
            \end{subfigure}
			\caption{A graphical example of portfolio compression across CCPs. Panel (a) exhibits a market consisting of market participants (i, j,and other circles) and CCPs (1 and 2). Panel (b) shows a multilateral compression solution to the market.}\label{fig-cross-ccp-compression}
		\end{figure}

		\begin{table}
			\renewcommand{\arraystretch}{0.9}
			\begin{tabular}{c|ccccc} \toprule  {} & Market 1 & Market 2 & Market 3 & Market 4 & Market 5 \\
				\midrule
				Original excess & 34.834 & 27.489 & 31.592 & 26.227 & 27.051 \\
				Conservative compression    & 0.924 & 0.906 & 0.954 & 0.927 & 0.911 \\
				efficiency & & & & & \\
				\midrule
				\multicolumn{6}{c}{\textit{A. Clearing without multilateral compression}}\\
				\midrule
				1 CCP           &    0.805  &    0.655  &    0.834 &     0.71 &    0.793 \\
				{}          	&    (0.0) & (0.0) & (0.0) & (0.0) & (0.0) \\
				2 CCPs          &   0.66 &  0.543 &  0.672 &  0.579 &  0.716 \\
				{}          	&   (0.056) &  (0.036) &  (0.048) &  (0.044) &  (0.025) \\
				3 CCPs          &  0.576 &  0.477 &  0.569 &  0.502 &  0.663 \\
				{}          	&  (0.055) &  (0.039) &  (0.047) &  (0.047) &  (0.024) \\
				4 CCPs          &  0.497&  0.407&  0.482 &  0.437 &  0.616 \\
				{}           	&  (0.062) & (0.042) & (0.055) & (0.049) & (0.028) \\
				5 CCPs          &  0.435 &   0.364 &   0.42 &  0.374 &  0.588 \\
				{}           	&  (0.059) &   (0.04) &   (0.055) &  (0.051) &  (0.029) \\
				6 CCPs          &  0.385 &  0.311 &  0.354 &  0.317 &  0.554 \\
				{}           	&  (0.056) &  (0.046) &  (0.058) &  (0.043) &  (0.025) \\
				7 CCPs          &   0.34 &  0.271 &  0.302 &   0.265 &   0.526 \\
				{}           	&   (0.057) &  (0.048) &  (0.055) &   (0.05) &   (0.03) \\
				8 CCPs          &  0.287 &   0.226 &  0.247 &  0.229 &  0.495 \\
				{}           	&  (0.058) &   (0.05) &  (0.057) &  (0.049) &  (0.028) \\
				9 CCPs          &  0.239 &  0.191 &  0.214 &   0.176 &  0.472 \\
				{}           	&  (0.059) &  (0.047) &  (0.051) &   (0.05) &  (0.028) \\
				10 CCPs         &  0.193 &  0.157 &  0.166 &   0.145 &  0.447 \\
				{}          	&  (0.055) &  (0.048) &  (0.056) &   (0.05) &  (0.028) \\
				\midrule
				\multicolumn{6}{c}{\textit{B. Clearing with multilateral compression}}\\
				\midrule
				1 CCP    		&    0.805&    0.655&    0.834&     0.71&    0.793\\
				{}   	 		&   (0.0) & (0.0) & (0.0) & (0.0) & (0.0) \\
				2 CCPs 			&  0.774 &  0.632 &  0.807 &  0.682 &  0.774 \\
				{}    			&  (0.021) &  (0.016) &  (0.024) &  (0.021) &  (0.014) \\
				3 CCPs			&  0.785 &  0.633 &   0.81 &  0.694 &  0.779 \\
				{}    			&  (0.022) &  (0.019) &   (0.022) &  (0.016) &  (0.012) \\
				4 CCPs 			&  0.794 &   0.638 &  0.816 &  0.695 &  0.787 \\
				{}    			&  (0.016) &   (0.02) &  (0.018) &  (0.017) &  (0.009) \\
				5 CCPs 			&  0.797 &   0.64 &   0.82 &    0.7 &   0.785 \\
				{}    			&  (0.011) &   (0.017) &   (0.017) &    (0.012) &   (0.01) \\
				6 CCPs 			&  0.798 &  0.644 &  0.826 &  0.702 &  0.787 \\
				{}    			&  (0.012) &  (0.014) &  (0.012) &  (0.013) &  0(0.008) \\
				7 CCPs 			&  0.798 &  0.646 &  0.825 &  0.705 &  0.789 \\
				{}    			&  (0.011) &  (0.012) &  (0.011) &  (0.008) &  (0.006) \\
				8 CCPs 			&  0.802 &  0.648 &  0.827 &  0.705 &  0.788 \\
				{}    			&  (0.006) &  (0.011) &  (0.011) &  (0.008) &  (0.007) \\
				9 CCPs 			&  0.802 &  0.651 &   0.829 &  0.705 &  0.788 \\
				{}    			&  (0.007) &  (0.007) &   (0.01) &  (0.009) &  (0.007) \\
				10 CCPs 		&  0.802 &  0.648 &  0.829 &  0.706 &   0.786 \\
				{}   			&  (0.007) &  (0.011) &  (0.008) &  (0.008) &   (0.01) \\
				\bottomrule
			\end{tabular}
			\caption{Effect of bilateral compression with an increasing number of CCPs. Columns report the average efficiency ratio and standard deviation in parentheses for the 5 markets with the largest notional amounts outstanding on April 15, 2016. Panel A. only applies bilateral compression with CCPs. Panel B. includes multilateral compression across CCPs.\label{tab:ccp-compression-profileration}}
		\end{table}

		Table \ref{tab:ccp-compression-profileration} reports the results of the exercise for the five markets with the largest aggregate notional on the last day of our time window. The table shows the results for the two scenarios. First, we observe that an increase in the number of CCPs has vast adverse effects on the elimination of excess. The proliferation of CCPs reorganizes the web of obligation creating separated segments around each CCP. Global netting opportunities are dramatically lost at the bilateral level. Whereas the single CCP configuration was sufficiently efficient to compensate for the duplication of aggregate notional, this balance does not hold once the number of CCPs increase.

		Second, we find that the adverse effect of proliferation is almost entirely offset when obligations can be compressed multilaterally across CCPs. Netting opportunities are recovered once the compression exercise includes several CCPs. In particular, proliferation beyond two CCPs yields levels very close to the single CCP scenario.

		Note that we have assumed a uniform distribution of trades among CCPs which entails equivalent market shares. In general, increasing concentration to some CCPs should reduce the adverse effects. Nevertheless, the results on cross CCP compression would still hold qualitatively.

	\subsection{Discussion}	 

		Mandates and increased incentives to clear are at the heart of the regulatory response to the GFC. Central clearing and portfolio compression are both post-trade technologies which have reshaped the organization of OTC markets. However their interplay has been so far unclear. We provide here a simple intuition. CCPs provide natural netting opportunities. Yet, they also duplicate the aggregate notional in the market. While exposures towards CCPs are admittedly of a different nature than OTC exposures, concerns about risk concentration and resilience have been raised (see \cite{duffie2011does}). In turn, the proliferation of CCPs has brought several concerns from interoperability and cross-border issues to losses in netting efficiency. In this respect, we investigated the effect of central clearing on the aggregate notional amounts of OTC markets. 

		The results of our stylized exercise show that a proliferation of CCPs has adverse effects on netting opportunities. We have assumed that all CCPs practice bilaterally compress with their members, which is neither always the case in practice nor currently mandated. Furthermore, our findings show that multilateral compression across CCPs can almost entirely alleviate this concern. This result supports interoperability policies favoring the adoption of compression by CCPs and their mutual participation to multilateral cycles.

\section{Concluding remarks\label{sec:closing-remarks}}

	The post-crisis regulatory reforms have generated demand for new post-trade services such as portfolio compression in financial markets~\citep{fsb2017review}. This particular multilateral netting technique, which allows market participants to eliminate direct and indirect offsetting positions, has reportedly been responsible for the large downsizing of major OTC derivatives markets~\citep{aldasoro2018credit}.

	In this paper, we introduce a framework that empirically supports the large effects attributed to a market wide adoption of portfolio compression. We show that OTC markets with fungible obligations and counterparty risk generate large notional excess: gross volumes can far exceed the level required to satisfy every participants' net position. Dealers acting as intermediaries between customers but also between other dealers are the main determinant for the levels of excess empirically observed in markets. 

	Using a granular dataset on bilateral obligations resulting from CDS contracts, we find that around $75\%$ of total market sizes is in excess, on average. Furthermore, we find that even when participants are conservative regarding their counterparty relationships, engaging in portfolio compression, on average, eliminates $85\%$ of the excess. Finally, we find that the loss of netting efficiency due to multiple CCPs can be offset when portfolio compression take place across CCPs.

	The large amounts of excess observed in markets can be a source of financial instability, in particular in times of crisis (\cite{cecchetti2009central,acharya2014counterparty}). Given the empirical structure of OTC markets, portfolio compression can eliminate most of the excess even under conservative constraints. Furthermore, the efficiency of multilateral compression can, on its own, explain the large reductions in size historically observed in some OTC markets since the GFC. Therefore, policy discussions regarding the activity of OTC markets would benefit from detailed information related to clients use of portfolio compression. For instance, a rapid reduction in gross volumes due to portfolio compression can increase liquidity by reducing inventory costs as suggested by \cite{duffie2018post}. In contrast, a reduction driven by participants exiting the market would result in a decrease of liquidity provision. While the latter represents a change in economic dimension, the former originates from a change in the accounting dimension. Each mechanism implies a different assessment of market liquidity. By this token, the results from this paper also highlight the importance of granular data to gain economic insights from markets dynamics, as put forward by several recent policy initiatives \citep{draghi2016speech,coeure2017speeach}.
	


	The use of portfolio compression to mitigate the adverse netting effect of a proliferation in the number of CCPs shows that the combination of these post-trade services can prevent the current high levels of risk concentration that a one-single-CCP scenario would imply. More in general, compressing out excessive positions in times of distress could efficiently limit both real and expected propagations of shocks. 

	To the best of our knowledge, this work is the first to propose a comprehensive framework to analyze the mechanics of compression in terms of both feasibility and efficiency. The extent to which increasing demand for post-trade services in response to regulatory reforms affects market monitoring, market micro-structure and financial stability is unclear. Avenues for future work include empirical assessments of the effect of portfolio compression adoption on trading strategies, liquidity provision and market segmentation of OTC markets.

\bibliographystyle{apalike}
\bibliography{reference}

\newpage

    \section{Proofs\label{app-proofs}}

\subsection{Proposition \ref{prop:minimum-notional}}
	\begin{proof}
	The proof consists of two steps. 
    
    \begin{enumerate}
    
    \item First, we show that given a market $G=(N,E)$, we can always find a net-equivalent market $G'$ with total notional of $x'$ as in Equation \ref{eq:minnotionalnetequivalent}.
    
    
    Consider the  partition of $N$ into the following disjoint subsets: $N^+ = \{i | v^{net}_i > 0\}, N^- = \{i | v^{net}_i < 0\}$ and $N^0 = \{i | v^{net}_i = 0\}$ (such that $N = N^+ \bigcup N^- \bigcup N^0$). Let $B \in N \times N$ be a new set of edges (each with weight $b_{ij}$) such that:
    \begin{itemize}
    	\item $\forall b_{ij}$ s.t. $(i, j)\in B$, $i\in N^+, j\in N^-$;
    	\item $\sum_j b_{ij} = v^{net}_i, ~~~ \forall i \in N^+$;
        \item $\sum_i b_{ij} = v^{net}_j, ~~~ \forall j \in N^-$.
    \end{itemize}
    
	The total notional of the market $G'=(N,B)$ is thus given by:
    \begin{equation}
    	x' = \sum_i \sum_j b_{ij} = \sum_{i\in N^+} v^{net}_i = \sum_{i\in N^-} |v^{net}_i|. \nonumber
    \end{equation}
    As edges in $B$ only link two nodes within $N$ (i.e., the system is closed), the sum of all net position is equal to 0: $\sum_i v^{net}_i = 0$. Hence, we have: $\sum_{i \in N^+} v^{net}_i + \sum_{j \in N^-} v^{net}_j = 0$. We see that, in absolute terms, the sum of net positions of each set ($N^+$ and $N^-$) are equal: $|\sum_{i \in N^+} v^{net}_i| = |\sum_{j \in N^-} v^{net}_j|$. As all elements in each part have the same sign by construction, we obtain: $\sum_{i \in N^+} |v^{net}_i| = \sum_{j \in N^-} |v^{net}_j|$.  As a result, we have:  $\sum_{i\in N^+} v^{net}_i = \frac{1}{2}|\sum_{i \in N} v^{net}_i|$.\\
    
	\item Second, we show that $x'$ is the minimum total notional attainable from a net-equivalent operation over $G=(N,E)$. We proceed by contradiction. Consider $G'=(N,B)$ as defined above and assume there exists a $G^*=(N,B^*)$ defined as a net-equivalent market to $G'$ such that $x^*<x'$. At the margin, such result can only be obtained by a reduction of some weight in $B$: $\exists b^*_{ij} < b_{ij}$. If $x^*<x'$, then there exists at least one node for which this  reduction is not compensated and thus $\exists v^{*net}_i < v^{net}_i$. This violates the net-equivalent condition. Hence, $x'=\sum_{i: \;\; v_i^{net} > 0}^n  v_i^{net}$ is the minimum net equivalent notional.
    
    \end{enumerate}

		\end{proof}

\bigskip

\subsection{Lemma \ref{lem:iff-excess}}

\begin{proof}
	By definition, $\delta(i) = 1 \Leftrightarrow \sum_{j} e_{ij} \cdot \sum_{j} e_{ji} >0$: a dealer has thus both outgoing and incoming edges. Then it holds that: 
    \begin{equation}
    	\delta(i) = 1 \;\;\Rightarrow  \;\;v^{gross}_i > |v^{net}_i| \;\; \Leftrightarrow \;\; 	\sum_{j} e_{ij} + \sum_{j} e_{ji}> \left|\sum_{j} e_{ij} - \sum_{j} e_{ji}\right|. 
        \nonumber
    \end{equation}
    
\noindent In contrast, for a customer  $\sum_{j} e_{ij} \cdot \sum_{j} e_{ji} = 0$ and thus $\delta(i) = 0$. Then it holds that:

\begin{equation}
    	\delta(i) = 0  \;\; \Rightarrow v^{gross}_i = |v^{net}_i| \;\; \Leftrightarrow  \;\; \sum_{j} e_{ij} + \sum_{j} e_{ji} = \left|\sum_{j} e_{ij} - \sum_{j} e_{ji}\right|. \nonumber
    \end{equation}

\noindent   The equality is simply proven by the fact that if $i$ is a customer selling (resp. buying) in the market, then $\sum_{j} e_{ji}=0$ (resp. $\sum_{j} e_{ij}=0$) and thus both ends of the above equation are equal.\\
    
\noindent
    If $G=(N,E)$ has $\sum_{i\in N} \delta(i)=0$, then all market participants are customers, and we thus have: $v^{gross}_i = |v^{net}_i| ~ \forall i \in N$. As a result, the excess is given by
    
 \begin{equation}
    	\Delta(G) = x - \frac{1}{2}\sum_{i} \left|v^{net}_i\right| = x - \frac{1}{2}\sum_{i} \left|v^{gross}_i\right|.    \nonumber
    \end{equation}

\noindent 
As in the proof of Proposition 1, the market we consider is closed (i.e., all edges relate to participants in $N$) and thus: $\sum_{i} |v^{gross}_i| = 2x$. We thus have no excess in such market: $\Delta(G) = 0$.\\


\noindent If $G=(N,E)$ has $\sum_{i\in N} \delta(i)>0$, then some market participants have $v^{gross}_i > |v^{net}_i|$. As a result, the excess is given by:
    \begin{align}
    	\Delta(G) & = x - \frac{1}{2}\sum_{i} |v^{net}_i| = \frac{1}{2}\sum_{i} |v^{gross}_i| - \frac{1}{2}\sum_{i} |v^{net}_i| =\nonumber \\ 
        & = \sum_{i} |v^{gross}_i| - \sum_{i} |v^{net}_i| > 0 \nonumber
    \end{align}
\end{proof}

\bigskip
\subsection{Proposition \ref{prop:additivity-excess}}
    
    

    

\begin{proof}
    For sake of clarity, in the following we only focus the notation on the set of edges for the computation of excess. In general, let us decompose the set of edges $E$ in two subsets $E^1$ and $E^2$ such that $E = E^1 \cup E^2$ and . 
    
    \begin{equation}
	e_{ij} =  e^1_{ij} + e^2_{ij}, \forall (i, j) \in E. \label{eq:splitmatrix}
    \end{equation}
    
    We compute the excess for the matrix $e_{ij}$:
    
    \begin{equation}
    	\sum_{ij} e_{ij} - 0.5 \sum_i \left|\sum_j (e_{ij}-e_{ji})\right|. \label{eq:splitexcess}
    \end{equation}
    
    Expanding and substituting \ref{eq:splitmatrix} into \ref{eq:splitexcess}, we obtain: 

   \begin{align}
    \sum_{ij} e_{ij} - 0.5 \sum_i \left|\sum_j (e_{ij}-e_{ji})\right| & = \sum_{ij} (e^1_{ij} + e^2_{ij}) + \nonumber\\ &- 0.5 \sum_i \left|\sum_j (e^1_{ij} - e^1_{ji} + e^2_{ji} - e^2_{ji})\right| =\nonumber \\
	& = \sum_{ij} e^1_{ij} +  \sum_{ij} e^2_{ij} + \nonumber \\ & - 0.5 \sum_i \left|\sum_j (e^1_{ij} - e^1_{ji}) + \sum_j( e^2_{ji} - e^2_{ji})\right| \label{eq:excessdecomp}
    \end{align}

By Jensen's inequality, we have that:
\begin{equation}
\left|\sum_j (e^1_{ij} - e^1_{ji}) + \sum_j( e^2_{ji} - e^2_{ji})\right| \leq  \left|\sum_j (e^1_{ij} - e^1_{ji})\right| + \left|\sum_j( e^2_{ji} - e^2_{ji})\right|\nonumber	
\end{equation}

therefore from \ref{eq:excessdecomp} it follows that:

\begin{align}
\sum_{ij} e_{ij} - 0.5 \sum_i \left|\sum_j (e_{ij}-e_{ji})\right| & \geq \sum_{ij} e_{ij}^1 - 0.5 \sum_i \left|\sum_j e_{ij}^1 - e_{ji}^1\right| +\nonumber \\
& + \sum_{ij} e^2_{ij} - 0.5 \sum_i \left|\sum_j e_{ij}^2 - e_{ji}^2\right| \nonumber
\end{align}

which proves the claim.

\noindent We now identify specific cases under our framework in which the relationship holds. Let us decompose the original additivity expression:
    \begin{align}
    	\Delta(E) &= \Delta(E^D) + \Delta(E^C) \nonumber\\
    	\sum_i | \sum_j(e_{ij} - e_{ji})| &= \sum_i | \sum_j(e^D_{ij} - e^D_{ji})| +  \sum_i | 	\sum_j(e^C_{ij} - e^C_{ji})| \nonumber
    \end{align} 

    We can decompose each part in the context of a dealer-customer network.\\

    1) For the whole network we have
    \begin{align}
    \sum_i | \sum_j(e_{ij} - e_{ji})| &= \sum^{dealer}_d | \sum_j(e_{dj} - e_{jd})| + \sum^{customer}_c | \sum_j(e_{cj} - e_{jc})| \nonumber \\
    &=\sum^{dealer}_d | \sum_j(e_{dj} - e_{jd})| + \sum^{customer^+}_{c^+} | \sum_j(e_{c^+j} - e_{jc^+})| + \nonumber \\
    & +  \sum^{customer^-}_{c^-} | \sum_j(e_{c^-j} - e_{jc^-})| =  \nonumber  \\
    &= \sum^{dealer}_d | \sum_j(e_{dj} - e_{jd})| + \sum^{customer^+}_{c^+} | \sum_j(e_{c^+j})| + \nonumber \\
    & + \sum^{customer^-}_{c^-} | \sum_j(- e_{jc^-})| = \nonumber \\
    &= \sum^{dealer}_d | \sum_j(e_{dj} - e_{jd})| + \sum^{customer^+}_{c^+} \sum^{dealer}_de_{c^+d} + \sum^{customer^-}_{c^-} \sum^{dealer}_de_{dc^-} \nonumber 
    \end{align}

    2) For the dealer network we have
    \begin{align}
    \sum_i | \sum_j(e^D_{ij} - e^D_{ji})| &= \sum^{dealer}_d | \sum^{dealer}_h(e^D_{dh} - e^D_{hd})| \nonumber 
    \end{align}

    3) For the customer network we have
    \begin{align}
    \sum_i | \sum_j(e^C_{ij} - e^C_{ji})| &= \sum^{dealer}_d | \sum_j(e^C_{dj} - e^C_{jd})| + \nonumber \\ 
    & + \sum^{customer^+}_{c^+} | \sum_j(e^C_{c^+j} - e^C_{jc^+})| \sum^{customer^-}_{c^-} | \sum_j(e^C_{c^-j} - e^C_{jc^-})| \nonumber\\
    &= \sum^{dealer}_d | \sum_j(e^C_{dj} - e^C_{jd})| + \sum^{customer^+}_{c^+} \sum^{dealer}_d e^C_{c^+d} + \sum^{customer^-}_{c^-} \sum^{dealer}_de^C_{dc^-}\nonumber
    \end{align}

    Combining equations, we obtain:
    \begin{align}
    \sum^{dealer}_d | \sum^{n}_j(e_{dj} - e_{jd})| = \sum^{dealer}_d | \sum^{dealer}_h(e^D_{dh} - e^D_{hd})| + \sum^{dealer}_d | \sum^{customer}_c(e^C_{dc} - e^C_{cd})|  \nonumber
    \end{align}

    We continue decomposing the different elements. \\

    1) For the whole network:

    \begin{align}
    \sum^{dealer}_d | \sum^{n}_j(e_{dj} - e_{jd})| &= \sum^{dealer}_d | \sum^{dealer}_h(e_{dh} - e_{hd}) + \sum^{customer^+}_{c^+}(e_{dc^+} - e_{c^+d}) + \sum^{customer^-}_{c^-}(e_{dc^-} - e_{c^-d})| \nonumber\\
    &= \sum^{dealer}_d | \sum^{dealer}_h(e_{dh} - e_{hd}) + \sum^{customer^+}_{c^+}e_{dc^+} - \sum^{customer^-}_{c^-}e_{c^-d}|\nonumber
    \end{align}

    2) for the dealer and customer networks:

    \begin{align}
    \sum^{dealer}_d | \sum^{dealer}_h(e^D_{dh} - e^D_{hd})| + \sum^{dealer}_d | \sum^{customer}_c(e^C_{dc} - e^C_{cd})| &=\nonumber\\ 
    \sum^{dealer}_d | \sum^{dealer}_h(e^D_{dh} - e^D_{hd})| + &\sum^{dealer}_d |\sum^{customer^+}_{c^+}e^C_{dc^+} - \sum^{customer^-}_{c^-}e^C_{c^-d}| \nonumber
    \end{align}

    After this decomposition, we can remove the subscripts related to the different networks, and we obtain the general condition for additive excess:

    \begin{align}
    \sum^{dealer}_d | \sum^{dealer}_h(e_{dh} - e_{hd}) + \sum^{customer^+}_{c^+}e_{dc^+} - \sum^{customer^-}_{c^-}e_{c^-d}| &= \nonumber\\ 
    \sum^{dealer}_d | \sum^{dealer}_h(e_{dh} - e_{hd})| + &\sum^{dealer}_d |\sum^{customer^+}_{c^+}e_{dc^+} - \sum^{customer^-}_{c^-}e_{c^-d}| \nonumber
    \end{align}

Hence, the above relationship holds when
 \begin{enumerate}
    	\item $\sum^{dealer}_h(e_{dh} - e_{hd}) = 0, ~~~ \forall d\in D$ \\
		\item[]or
        \item $\sum^{customer^+}_{c^+}e_{dc^+} - \sum^{customer^-}_{c^-}e_{c^-d} = 0, ~~~ \forall d \in D$
    \end{enumerate}
\end{proof}

\bigskip
\subsection{Proposition \ref{prop:non-conservative-feasibility}}
\begin{proof}
	Non-conservative compression tolerances allow all possible re-arrangements of edges. Hence, the only condition for non-conservative compression to remove excess (i.e., $\Delta^{c{^n}}_{\text{red}}(G)>0$) is merely that excess is non-zero (i.e., $\Delta(G) > 0$). From Lemma 1, we know that positive excess exists in $G=(N,E)$ only when there is intermediation (i.e., $\exists i \in N | \delta(i) =1$).
\end{proof}

\subsection{Proposition \ref{prop:non-conservative-efficiency}}
\begin{proof}
	We proceed by defining a procedure that respects the non-conservative compression constraints and show that this procedure (algorithm) generates a new configuration of edges such that the resulting excess is 0.\\
    

\noindent     Similar to the proof of Proposition 1, consider the three disjoint subsets $N^+ = \{i | v^{net}_i > 0\}, N^- = \{i | v^{net}_i < 0\}$ and $N^0 = \{i | v^{net}_i = 0\}$, such that $N = N^+ \bigcup N^- \bigcup N^0$. Let $B$ be a new set of edges such that:
    \begin{itemize}
    	\item $\forall b_{ij}\in B, ~~~ i\in N^+, j\in N^-$
    	\item $\sum_j b_{ij} = v^{net}_i, ~~~ \forall i \in N^+$
        \item $\sum_i b_{ij} = v^{net}_j, ~~~ \forall j \in N^-$
    \end{itemize}
   	
	The market $G'=(N,B)$ is net-equivalent to $G$ while the total gross notional is minimal in virtue of Proposition 1. The nature of the new edges makes $G'$ bipartite (i.e., $\forall b_{ij}\in B, ~~~ i\in N^+, j\in N^-$), hence, there is no intermediation in $G'$. The procedure depicted above to obtain $B$ is a meta-algorithm as it does not define all the steps in order to generate $B$. As a result, several non-conservative compression operation $c^n$ can satisfy this procedure. Nevertheless, by virtue of Proposition 3, each of these non-conservative compression operation lead to $\Delta^{c^n}_{res}(G) = \Delta(G') = 0$
\end{proof}

\subsection{Proposition \ref{prop:conservative-feasibility}}
    
\begin{proof}
	In a conservative compression, we have the constraint:
    \[
    	0 \leq e'_{ij} \leq e_{ij} ~~~ \forall i,j \in N
    \]
    
   \noindent At the individual level, assume $i$ is a customer selling in the market (i.e., $\delta(i) = 0$). Under a conservative approach, it is not possible to compress any of the edges of $i$. In fact, in order to keep the net position of $i$ constant, any reduction of $\varepsilon$ in an edge of $i$ (i.e., $e'_{ij} = e_{ij} - \varepsilon$) requires a change in some other edge (i.e., $e'_{ik} = e_{ik} + \varepsilon$) in order to keep $v'^{net}_i =v^{net}_i $. Such procedure violates the conservative compression tolerance: $e'_{ik} = e_{ik} + \varepsilon > e_{ik}$. The same situation occurs for customers buying. Conservative compression can thus not be applied to node $i$ if $\delta(i) = 0$.\\
 
 \noindent
    The only configuration in which a reduction of an edge $e_{ij}$ does not require a violation of the conservative approach and the net-equivalence condition is when $i$ can reduce several edges in order to keep its net balance. In fact, for a node $i$, the net position is constant after a change $\sum_j e'_{ij} =  \sum_j e_{ij} - \varepsilon$ if it is compensated by a change $\sum_j e'_{ji} =  \sum_j e_{ji} - \varepsilon$. Only dealers can apply such procedure. Furthermore, such procedure can only be applied to links with other dealers: a reduction on one link triggers a cascade of balance adjusting that can only occur if other dealers are concerned as customers are not able to re-balance their net position as shown above. Hence, the redundant excess for a conservative approach emerges from intra-dealer links. \\
 \noindent   
    Finally, the sequence of rebalancing and link reduction can only stop once it reaches the initiating node back. Hence, conservative compression can only be applied to directed closed chains of intermediation, that is, a set of links $E^* \subset E$ such that all links have positive values $\prod_{e^* \in E^*} e^* > 0$.
\end{proof}

\bigskip
\subsection{Proposition \ref{prop:conservative-dealer-customer-efficiency}}
\begin{proof}
	From Proposition \ref{prop:conservative-feasibility}, we know that the conservative compression approach can only reduce excess in closed chains of intermediation. 
	Given a market $G=(N,E)$, let $i \in N$ satisfy the following condition:
	\[
		\begin{cases}
			\sum_j e^{C}_{ij} >0, ~~ e^{C}_{ij} \in E^{C}  \\
			\sum_j e^{C}_{ji} >0, ~~ e^{C}_{ji} \in E^{C} 
		\end{cases}	
	\]
	The participant $i$ is thus a dealer in the market. More precisely, irrespective of her activity with other dealers (i.e., intra-dealer market $E^D$), $i$ interacts with customers both on the buy and on the sell side. By definitions, those sets of counterparties generate no excess as they are only active on one side. 

	As a result, $i$ belongs to open chains of intermediation where customers selling are on the sender's end of the chain while customers buying are on the receiver's end of the chain. By virtue of the conservative setting, it is not possible to compress those open chain as both extreme-ends of the chains are not intermediaries. In turn, the excess generated by those chains cannot be compressed:  $\Delta^{c{^c}}_{\text{res}}(G) > 0$.

	Assume, instead, that all dealers only interact with one type of customer:
	\[
		\begin{cases}
			\sum_j e^{C}_{ij} \cdot \sum_j e^{C}_{ji}  = 0, ~~ e^{C}_{ij}, e^{C}_{ji} \in E^{C}  \\
			\sum_j e^{D}_{ik} \cdot \sum_j e^{D}_{ki} \geq 0, ~~ e^{D}_{ik}, e^{D}_{ki} \in E^{D}  
		\end{cases}	
		~~~ \forall i \in N^D
	\]
	In such case, there always exists a configuration of the intra-dealer market such that all the excess can be removed via conservative compression. In fact, if the intra-dealer market is composed of equally weighted closed chains of intermediation, they can all be conservatively compressed out of the market. As a result, only dealer-customer trades with remain. Given the original configuration, no dealer would be intermediating anymore and no excess would be left in the market after consevative compression.

	We thus see that in order to ensure positive residual excess after conservative compression, we need open chains of intermediation in the original market which are ensured by the existence of direct intermediation between customers.
\end{proof}

\bigskip
\subsection{Lemma \ref{lemm:conservative-closed-chain}}
\begin{proof}
	A conservative compression on a closed chain of intermediation $K=(N,E)\rightarrow(K,E')$ implies that, in order for the compression to be net equivalent (i.e., $v^{'net}_i = v^{net}_i ~~ \forall i \in N$), a reduction by and arbitrary $\varepsilon \in [0, \max_{ij}\{e_{ij} \text{ s.t. } (i, j) \in E\}] $ on an edge $e'_{ik} = e_{ik} - \varepsilon$ must be applied on all other edges in the chain: $e' = e - \varepsilon ~~ \forall e' \in E'$.\\
    
  \noindent
    Overall, reducing by $\varepsilon$ one edge, leads to an aggregate reduction of $|E|\varepsilon$ after re-balancing of net positions.\\

\noindent
    Recall that, in a conservative compression, we have $0 \leq e'_{ij} \leq e_{ij}$. Hence, for each edge, the maximum value that $\varepsilon$ can take is $e_{ij}$. At the chain level, this constraint is satisfied i.f.f. $\varepsilon = min_e\{E\}$. The redundant excess is given by $|E|min_e\{E\}$ and the residual excess is thus
    \begin{equation}
    	\Delta^c_{res}(K) = \Delta(K) - |E|min_e\{E\} \nonumber
    \end{equation} \end{proof}


\bigskip
\subsection{Proposition \ref{prop:hybrid-efficiency}}

\begin{proof}
  If $\Delta(N,E) = \Delta(N,E^D) + \Delta(N,E^C)$, then we can separate the compression of each market.\\
  
  \textbf{Intra-dealer market ($N,E^D$).} According to the hybrid compression, the set of constraints in the intra-dealer market is given by a non-conservative compression tolerances set. According to Proposition 4, the residual excess is zero. We thus have:
  \begin{equation}
  	\Delta^{c^h}_{res}(N,E^D) = 0. \nonumber
  \end{equation}
  
  \textbf{Intra-dealer market ($N,E^D$).} According to the hybrid compression, the set of constraints in the customer market is given by a conservative compression tolerances set. Since, by construction, the customer market does not have closed chains of intermediation, it is not possible to reduce the excess on the customer market via conservative compression. We thus have:
    \begin{equation}
  	\Delta^{c^h}_{res}(N,E^C) = \Delta(N,E^C). \nonumber
  \end{equation}
  \noindent
  Finally, we obtain
  \begin{align*}
  				\Delta^{c^h}_{\text{res}}(N,E) &= \Delta^{c^h}_{\text{res}}(N,E^D) + \Delta^{c^h}_{\text{res}}(N,E^C)\\
                &= \Delta(N,E^C)
  \end{align*} \end{proof}

\subsection{Proposition \ref{prop:bilateral-feasibility}}

	\begin{proof}
		If the market $G=(N,E)$ is such that $\nexists i,j \in N ~~~ s.t.  ~~~ e_{ij}.e_{ji} > 0 ~~ \mbox{where} ~~ e_{ij},e_{ji}\in E$ then the compression tolerances will always be:
		\[
			a_{ij} = b_{ij} = \max{\{e_{ij}-e_{ji},0\}} = e_{ij}
		\]
		Hence, $\Delta^{c^b}_{red}(G) = \Delta_{red}(G)$ and thus $\Delta^{c^b}_{res}(G)=0$.
		If the market $G=(N,E)$ is such that $\exists i,j \in N ~~~ s.t.  ~~~ e_{ij}.e_{ji} > 0 ~~ \mbox{where} ~~ e_{ij},e_{ji}\in E$ then the bilateral compression will yield a market $G'=(N,E')$ where $x^{'} < x$. Hence,  $\Delta^{c^b}_{red}(G) < \Delta_{red}(G)$ and thus $\Delta^{c^b}_{res}(G)>0$
	\end{proof}

\subsection{Proposition \ref{prop:bilateral-efficiency}}
		\begin{proof}
			If the market $G=(N,E)$ is such that $\exists i,j \in N ~~~ s.t.  ~~~ e_{ij}.e_{ji} > 0 ~~ \mbox{where} ~~ e_{ij},e_{ji}\in E$, then, bilaterally compressing the pair $i$ and $j$ yields the following situation.
			Before compression, the gross amount on the bilateral pair was $e_{ij} + e_{ji}$. After compression, the gross amount on the same bilateral pair is $|e_{ij} - e{ji}|$. Hence, we have a reduction of gross notional of $2.min\{e_{ij},e_{ji}\}$. The market gross notional after compression of this bilateral pair is thus given by: $x^{'} = x - 2.min\{e_{ij},e_{ji}\}$ and the excess in the new market (i.e., residual excess after having bilaterally compressed the pair $(i,j)$) follows the same change: $\Delta_{res}(G) = \Delta(G) -  2.min\{e_{ij},e_{ji}\}$.
			We generalize the result by looping over all pairs and noting that the reduction $min\{e_{ij},e_{ji}\}$ is doubled counted: pairing by $(i,j)$ and $(j,i)$. Hence, we reach the following expression of the residual excess: 
			\[
				\Delta^{c^b}_{res}(G) = \Delta(G) - \sum_{i,j \in N} \min \{e_{ij}, e_{ji}\} ~~ \mbox{where} ~~ e_{ij},e_{ji}\in E \nonumber
			\]
		\end{proof}

\subsection{Proposition \ref{prop:efficiency-dominance}}

		\begin{proof}
			We proceed by analyzing sequential pairs of compression operators and show the pairing dominance before generalizing. We start with the bilateral compressor $c^b()$ and the conservative compressor $c^c()$. Let $(a^b_{ij}, b^b_{ij})\in\Gamma^b$ and $(a^c_{ij}, b^c_{ij})\in\Gamma^c$ be the set of compression tolerance for the bilateral and conservative compressor, respectively. We have the following relationship:
			\[
				a^c_{ij} \leq a^b_{ij} = b^b_{ij} \leq b^c_{ij}  ~~~ \forall e_{ij} \in E
			\]
			In fact, by definition of each compression tolerance set, we have:
			\[
				0 \leq \max\{e_{ij} - e_{ji}, 0\} \leq e_{ij}  ~~~ \forall e_{ij} \in E
			\]			
			Hence, we see that the set of possible values couple for bilateral compression is bounded below and above by the set of conservative compression values. By virtue of linear composition, a solution of the bilateral compression thus satisfies the conservative compression tolerance set. The other way is not true as the lower bound in the bilateral case $a^b_{ij}$ can be equal to $e_{ij} - e_{ji}$ while in the conservative case, we always have that $a^c_{ij} = 0$. Hence, in terms of efficiency, we have that a globally optimal conservative solution is always at least equal, in redundant excess, to the globally optimal bilateral solution: $\Delta^{c^b}_{red}(G) \leq \Delta^{c^c}_{red}(G)$. The case in which the efficiency of $\Delta^{c^c}_{red}(G)$ is higher is a function of the network structure of $G$. In fact, if the market $G$ only exhibits cycles of length one, we have $\Delta^{c^b}_{red}(G) = \Delta^{c^c}_{red}(G)$. Once $G$ exhibits higher length cycles, we have a strict dominance $\Delta^{c^b}_{red}(G) < \Delta^{c^c}_{red}(G)$.
			Similar reasoning is thus applied to the next pairing: conservative and hybrid compression tolerance sets. Let $(a^h_{ij}, b^h_{ij})\in\Gamma^h$ be the set of compression tolerance for the hybrid compressor. We have the following nested assembly:
			\[
				a^c_{ij} = a^h_{ij} ~~~ \mbox{and} ~~~ b^c_{ij} = b^h_{ij} 	~~~ \forall e_{ij} \in E^C
			\]
			\[
				a^c_{ij} = a^h_{ij} ~~~ \mbox{and} ~~~ b^c_{ij} \leq b^h_{ij} 	~~~ \forall e_{ij} \in E^D
			\]
			Where $E^C$ and $E^D$ are the customer market and the intra-dealer market, respectively, with $E^C+E^D = E$. 
			In fact, by definition of the compression tolerance sets in the customer market $E^C$ are the same while for the intra-dealer market we have:
			\[
				a^c_{ij} = a^h_{ij} = 0 ~~~ \mbox{and} ~~~ e_{ij} \leq + \infty 	~~~ \forall e_{ij} \in E^D
			\]
			Similar to the dominance between bilateral and conservative compression, we can thus conclude that: $\Delta^{c^c}_{red}(G) \leq \Delta^{c^h}_{red}(G)$. It is the relaxation of tolerances in the intra-dealer market that allows the hybrid compression to be more efficient than the conservative compression. By virtue of complementarity of this result, the hybrid and non-conservative pairing is straightforward: $\Delta^{c^h}_{red}(G) \leq \Delta^{c^n}_{red}(G)$. As we know from Proposition~\ref{prop:non-conservative-efficiency}, $\Delta^{c^n}_{red}(G) = \Delta(G)$, we thus obtain the general formulation of weak dominance between the 4 compression operators:
			\[
				\Delta^{c^b}_{red}(G) \leq \Delta^{c^c}_{red}(G) \leq\Delta^{c^h}_{red}(G) \leq\Delta^{c^n}_{red}(G) = \Delta(G)
			\]

		\end{proof}

\newpage
\section{A simple example with 3 market participants\label{app-example-3-market-participants}}

			\begin{figure}
				\centering
					\begin{tikzpicture}
						\begin{scope}[every node/.style={circle,thick,draw}]
							\node (i) at (0,0) {i};
							\node (j) at (0,3) {j};
							\node (k) at (3,1.5) {k};
						\end{scope}
						\begin{scope}[every node/.style={fill=white,circle},
						    every edge/.style={draw=red,very thick}]
						    \path [->] (i) edge node { 5} (j);
						    \path [->] (j) edge node {$10$} (k);
						    \path [->] (k) edge node {$20$} (i);
						\end{scope}
					\end{tikzpicture}
				\caption{Original configuration the market}\label{fig-example-initial}
			\end{figure}

			To better articulate the different ways in which portfolio compression can take place according to the conservative and non-conservative approach, let us take the following example of a market made of $3$ participants. Figure~\ref{fig-example-initial} graphically reports the financial network: $i$ has an outstanding obligation towards $j$ of notional value $5$ while having one from $k$ of notional value $20$ and $j$ has an outstanding obligation  towards $k$ of notional value $10$. For each participant, we compute the gross and net positions:

			\begin{align*}
				& v^{gross}_i = 25 ~~~~~ v^{net}_i = -15\\
				& v^{gross}_j = 15 ~~~~~ v^{net}_j = +5\\
				& v^{gross}_k = 30 ~~~~~ v^{net}_k = +10
			\end{align*}

			We also obtain the current excess in the market:

			\begin{equation*}
				\Delta(G) = 35 - 15 = 20
			\end{equation*}

			Let us first adopt a conservative approach. In this case, we can only reduce or remove currently existing bilateral positions. A solution is to remove the obligation between $i$ and $j$ and adjust the two other obligation accordingly (i.e., subtract the value of $ij$ from the two other obligations). The resulting market is represented in Figure~\ref{fig-example-compressed}(a). Computing the same measurements as before, we obtain:
			\begin{align*}
				& v'^{gross}_i = 15 ~~~~~ v'^{net}_i = -15\\
				& v'^{gross}_j = 5 ~~~~~~ v'^{net}_j = +5\\
				& v'^{gross}_k = 20 ~~~~~ v'^{net}_k = +10
			\end{align*}

			We also obtain the new excess in the market:

			\begin{equation*}
				\Delta^{cons}_{res}(G) = 20 - 15 = 5
			\end{equation*}

			We see that, after applying the conservative compression operator that removed the $(i, j)$ obligation, we have reduced the excess by $15$. It is not possible to reduce the total excess further without violating the conservative compression tolerances. We thus conclude that, for the conservative approach, the residual excess is $5$ and the redundant excess is $15$.

			Let us now go back to the initial situation of Figure~\ref{fig-example-initial} and adopt a non-conservative approach. We can now create, if needed, new obligations. A non-conservative solution is to remove trades and create 2: one going from $j$ to $i$ of value $5$ and one going from $k$ to $i$ of value $10$. We have created an obligation that did not exist before between $j$ and $i$. The resulting market is depicted in Figure~\ref{fig-example-compressed}(b). Computing the same measurements as before, we obtain:

			\begin{align*}
				& v'^g_i = 15 ~~~~~ v'^n_i = -15\\
				& v'^g_j = 5 ~~~~~~ v'^n_j = +5\\
				& v'^g_k = 10 ~~~~~ v'^n_k = +10
			\end{align*}

			We also obtain the current excess in the market:

			\begin{equation*}
				\Delta^{non-cons}_{res}(G) = 15 - 15 = 0
			\end{equation*}

			We observe that we have eliminated all excess in the resulting market while all the net positions have remained constant. Individual gross positions are now equal to the net positions. Nevertheless the solution has generated a new position (i.e., from $j$ to $i$). We thus conclude that, for the non-conservative approach, the residual excess is $0$ and the redundant excess is $20$.

			The results are summarized in Table \ref{tab:example-market}. Though simple, the above exercise hints at several intuitive mechanisms and results that are developed further in the paper. 

			\begin{figure}
				\centering
				\begin{subfigure}{0.4\textwidth}
					\begin{tikzpicture}
						\begin{scope}[every node/.style={circle,thick,draw}]
							\node (i) at (0,0) {$i$};
							\node (j) at (0,3) {$j$};
							\node (k) at (3,1.5) {$k$};
						\end{scope}
						\begin{scope}[every node/.style={fill=white,circle},
						    every edge/.style={draw=red,very thick}]
						    \path [->] (j) edge node {$5$} (k);
						    \path [->] (k) edge node {$15$} (i);
						\end{scope}
					\end{tikzpicture}
				\caption{After conservative compression}
				\end{subfigure}
				\smallskip
				\smallskip
				\begin{subfigure}{0.4\textwidth}
					\begin{tikzpicture}
						\begin{scope}[every node/.style={circle,thick,draw}]
							\node (i) at (1.5,0) {i};
							\node (k) at (3,3) {k};
							\node (j) at (0,3) {j};
						\end{scope}
						\begin{scope}[every node/.style={fill=white,circle},
						    every edge/.style={draw=red,very thick}]
						    \path [->] (j) edge node {$5$} (i);
						    \path [->] (k) edge node {$10$} (i);
						\end{scope}
					\end{tikzpicture}
				\caption{After non-conservative compression}
				\end{subfigure}
				\caption{Examples of conservative and non-conservative compression approaches.}\label{fig-example-compressed}
			\end{figure}      
            
            \begin{table}
              \centering
              \begin{tabular}{lcc}
               & \textbf{Conservative} & \textbf{Non-conservative}     \\
               \toprule
              Total excess & 20 & 20    \\
              Redundant excess & 15 & 20    \\
              Residual excess & 5 & 0    
              \end{tabular}
              \caption{Table summarizing the results applying conservative and non-conservative compression on the market with 3 participants in Figure \ref{fig-example-compressed}.}\label{tab:example-market}
             \end{table}

\newpage
\section{Further analysis on Conservative Compression\label{app-conservative-compression}}
				In order to reach a directed acyclic graph any algorithm would need to identify and break all closed chains of intermediation. Nevertheless, the sequences of chains to be compressed can affect the results. In fact, if two chains share edges, compressing one chain modifies the value of  obligations also present in the other one. There can be different values of residual excess depending on which closed chain is compressed first.

				Formally, we identify such case as a case of \textit{entangled chains of intermediation}.

				\begin{definition}[Entangled Chains] Two chains of intermediation, $K_1=(N_1,E_1)$ and $K_2=(N_2,E_2)$, are entangled if they share at least one obligation:
				\[
					E_1 \cap E_2 \neq \emptyset
				\]
				\end{definition}

				\begin{figure}
					\centering
						\begin{tikzpicture}
							\begin{scope}[every node/.style={circle,thick,draw}]
								\node (A) at (0,0) {A};
								\node (B) at (0,3) {B};
								\node (C) at (1.5,1.5) {C};
								\node (D) at (3.5,2.5) {D};
							\end{scope}
							\begin{scope}[every node/.style={fill=white,circle},
							    every edge/.style={draw=blue,very thick}]
							    \path [->] (A) edge node {$5$} (B);
							    \path [->] (B) edge node {$10$} (C);
							    \path [->] (C) edge node {$20$} (A);
							    \path [->] (C) edge node {$10$} (D);
							    \path [->] (D) edge node {$3$} (B);
							\end{scope}
						\end{tikzpicture}
						\caption{Example of market with entangled chains}\label{fig:entangled-chains}
				\end{figure}

				An illustration of entangled chains is provided in Figure~\ref{fig:entangled-chains} where the edge $BC$ is share by two chains of intermediation (i.e., $ABC$ and $BCD$). 

				As such, we formulate the following feature on a graph:

				\begin{definition}(Chain Ordering Proof). 
					A market is chain ordering proof w.r.t. to the conservative compression if the ordering of entangled chains by $\Phi$ does not affect the efficiency of compression.
				\end{definition}

				If the configuration of entangled chains is such that, according to the initial ordering of excess reduction resulting from a compression on each chain, the optimal sequence is not affected by the effects of compression on other entangled chains, the market is said to be chain ordering proof. Under the above Definition, the optimal conservative compression yields a Directed Acyclic Graph (DAG) where the excess is given by the following expression:

				\begin{proposition}\label{prop:conservative-efficiency}
				    Given a market $G=(N,E)$. If there are no entangled chains, we have:
				    \begin{equation}
				    	\Delta_{res}(G) = \Delta(G) - \sum_{K_i\in\Pi}\Phi(E_{K_i}) \nonumber
				    \end{equation}
				    In the presence of entangled chains, if $G=(N,E)$ is chain-ordering proof, we have
				    \begin{equation}
				    	\Delta_{res}(G) > \Delta(G) - \sum_{K_i\in\Pi}\Phi(E_{K_i}) \nonumber
				    \end{equation}
				    Where $\Pi$ is the set of all chains of intermediation in $G$.
				\end{proposition} 
				\begin{proof}
					\noindent
						If there are no entangled chains in $G=(N,E)$, then the following conservative procedure:
					    \begin{enumerate}
					    \item list all closed chains of intermediation $K_i \in \Pi$ and
					    \item maximally compress each chain separately,
					    \end{enumerate}
					    
					    \noindent reaches  maximal efficiency. The residual excess is given after aggregating the excess removed on each closed chain separately:
					\begin{equation}
					    	\Delta_{res}(G) = \Delta(G) - \sum_{K_i\in\Pi}|E_i|min_e\{E_i\}.\nonumber
					\end{equation}
					    If there are entangled chains but the market $G=(N,E)$ is chain ordering proof, compressing chains separately only provides the upper bound as there will be cases where entangled chains will need to be updated (due to the reduction of one or more edges). Hence, we have,
					\begin{equation}
					    	\Delta_{res}(G) > \Delta(G) - \sum_{K_i\in\Pi}|E_i| \min_e\{E_i\}. \nonumber
					    	\end{equation}
				\end{proof}
        
				For illustrative purpose, we present an algorithm that always reaches a global solution under the chain ordering proof assumption in the Appendix~\ref{app-algorithms}.

\newpage
\section{Compression Algorithms\label{app-algorithms}}
	\subsection{Non-Conservative Algorithm}

		In order to provide a rigorous benchmark, we propose a deterministic non-conservative compression algorithm that eliminates all excess. 

		\begin{algorithm}[H]
			\KwData{Original Market G=(N,E)}
			\KwResult{$G^*$ such that $\Delta_v(G^*) = 0$}
			Let $N^+ = \{s \mbox{~~s.t.~} v^s_n>0 \mbox{~~and~} s\in N\}$ be ordered such that $v^{net}_1>v^{net}_2$\;
			Let $N^- = \{s \mbox{~~s.t.~} v^{net}_s<0\mbox{~~and~} s\in N\}$ be ordered such that $v^{net}_1 >v^{net}_2$\;
			Let i = 1 and j = 1\;
			\While{$i != |N^+|$ and $j != |N^-|$}{
				Create edge $e^*_{ij}=min(v^{net}_i-\sum_{j'<j}e^*_{ij'},v^{net}_j-\sum_{i'<i}e^*_{i'j})$\;
				\If{$v^{net}_i = \sum_{j'<j}e^*_{ij'}$}
				{
					 $i = i + 1$\;
				}
				\If{$v^{net}_j = \sum_{i'<i}e^*_{i'j}$}
				{
					 $j = j + 1$\;
				}
			}
			\caption{A perfectly efficient non-conservative compression algorithm with minimal density}
		\end{algorithm}

		From the initial market, the algorithm constructs two sets of nodes $N^+$ and $N^-$ which contain nodes with positive and negative net positions, respectively. Note that nodes with 0 net positions (i.e., perfectly balanced position) will ultimately be isolated. They are thus kept aside from this point on. In addition, those two sets are sorted from the lowest to the highest absolute net position. The goal is then to generate a set of edges such that the resulting network is in line with the net position of each node. Starting from the nodes with the highest absolute net position, the algorithm generates edges in order to satisfy the net position of at least one node in the pair (i.e., the one with the smallest need). For example, if the node with highest net positive position is $i$ with $v^{net}_i$ and the node with lowest net negative position is $j$ with $v^{net}_j$, an edge will be created such that the node with the lowest absolute net positions does not need more edges to satisfy its net position constraint. Assume that the nodes $i$ and $j$ are isolated nodes at the moment of decision, an edge $e_{ij}=min(v^{net}_i,v^{net}_j)$ will thus be generated. In the more general case where $i$ and $j$ might already have some trades, we discount them in the edge generation process:  $e^*_{ij}=\min(v^{net}_i-\sum_{j'<j}e^*_{ij'},v^{net}_j-\sum_{i'<i}e^*_{i'j})$. The algorithm finishes once all the nodes have the net and gross positions equal.

		The characteristics of the market resulting from a compression that follows the above algorithm are the following

		\begin{quote}
			Given a financial network $G$ and a compression operator c() that is defined by the Algorithm 1, the resulting financial network $G_{min}=c(G)$ is defined as:
			\begin{equation}
				e_{ij} = 
				\begin{cases}
				    \min(v_n^i - \sum_{j'<j}e_{ij'},v_n^j - \sum_{i'<i}e_{i'j}),& \text{if } v_n^i \cdot v_n^j < 0\\
				    0,              & \text{otherwise} \nonumber
				\end{cases}
			\end{equation}
			where $i\in N^+ = \{s \mbox{~~s.t.~} v^s_n>0\}$ and $j\in N^-= \{s \mbox{~~s.t.~} v^s_n<0\}$.\\
			Moreover:
			\begin{itemize}
				\item $G_{min}$ is net-equivalent to $G$
				\item $\Delta_v(G_{min}) = 0$
			\end{itemize}
		\end{quote}
	\subsection{Conservative Algorithm}

		As we did for the non-conservative case, we now propose and analyze a conservative algorithm with the objective function of minimizing the excess of a given market with two constraints: (1) keep the net positions constant and (2) the new set of trades is a subset of the previous one.

		\begin{algorithm}[H]
			\KwData{Original Market G=(N,E)}
			\KwResult{$G^*$ such that $\Delta_v(G^*) < \Delta_v(G)$ and $E^* \in E$}
			Let $\Pi$ be set the of all directed closed chains in $G$\;
			Let $G^*=G$\;
			\While{$\Pi \neq \emptyset$}{
				Select $P=(N',E')\in \Pi$ such that $|N'|.min_{e \in E'} (e) = max_{P_i=(N'_i,E'_i) \in \Pi} (|N'_{P_i}|.min_{e \in E'_{P_i}} (e)))$\;
				$e_{ij} = e_{ij} - min_{e \in E'} (e)$ for all $e_{ij}\in E'$\;
				$E^*=E^*\setminus \{e: e = min(E')\}$\;
				$\Pi \setminus \{P\}$
			}
			\caption{A deterministic conservative compression algorithm}
		\end{algorithm}

		$~$\\

		The algorithm works as follows. First, it stores all the closed chains present in the market. Then, it selects the cycle (i.e., closed chain) that will result in the maximum marginal compression (at the cycle level), that is, the cycle where the combination of the number of nodes and the value of the lowest trades is maximized. From that cycle, the algorithm removes the trade with the lowest notional and subtracts this value from the all the trades in the cycle. It then removes the cycle from the list of cycles and iterates the procedure until the set of cycles in the market is empty.

		At each cycle step $t$ of the algorithm, the excess of the market is reduced by:
		\[
			\Delta_t = \Delta_{t-1} - |N'|min_{e \in E'} (e)
		\]

		At the end of the algorithm, the resulting compressed market does not contain directed closed chains anymore: it is a Directed Acyclic Graph (DAG). Hence no further conservative compression can be applied to it.

\newpage
\section{Programming characterization and optimal algorithm\label{app-programming-charac}}
	
	\subsection{Programming formalization}

		Compression can be seen as the solution of a mathematical program which minimizes a non-decreasing function of gross notional under given net-positions. By introducing constraints on counterparty relationships, we will recover the hybrid and conservative compression.

		In particular, let $E\prime$ denote the set of edges after compression and let $f: E\prime \rightarrow \mathbb{R}$ be  a non decreasing function, the general compression problem is to find the optimal set $e\prime_{ij}$ in the following program: 

		\begin{problem}[General compression problem]
		\label{prob:general}
		\begin{equation}
		\label{eq:lpproblem1}
		\begin{array}{lll}
		\mathrm{min} & f(E\sp{\prime})  \\
		& &\\
		\text{s.t.}&   \sum_j \left(e\sp{\prime}_{ij} - e\sp{\prime}_{ji}\right) = v_i, \forall i \in V &\text{ [net position constraint]}\\ & \\
		& a_{ij} \leq e_{ij}\sp{\prime} \leq b_{ij}, \forall (i, j) \in E^\prime \subseteq (N \times N) & \mbox{[compression tolerances]}\\
		\end{array}
		\nonumber
		\end{equation}
		\end{problem}

		\noindent
		with  $a_{ij} \in [0, \infty)$ and $b_{ij} \in [0, \infty]$. We will refer to $E\sp{\prime}$ as the vector of solutions of the problem.

		Problem \ref{prob:general} maps all the compression types by translating the compression tolerances (counterparty constraints) and adopting a specific functional form for $f$. As we are interested in reducing the total amount of notional, we will set $f(E\sp{\prime}) = \sum_{ij} e_{ij}\sp{\prime}$.
		The non-conservative compression problem is obtained by setting $e_{ij} \in [0, \infty)$, as follows:

		\begin{problem}[Non-conservative compression problem]
		\label{prob:nonconservative}
		\begin{equation}	
		\begin{array}{lll}
		\mathrm{min} & \sum_{ij} e^{\prime}_{ij}  \\
		& &\\
		\text{s.t.}&   \sum_j \left(e^{\prime}_{ij} - e^{\prime}_{ji}\right) = v_i, \;\; \forall i \in N &\\ & \\
		&  e_{ij}^{\prime} \in [0, \infty), \forall (i, j) \in N \times N & \\
		\end{array}
		\nonumber
		\end{equation}
		\end{problem}

		In problem \ref{prob:nonconservative} the tolerances are set to the largest set possible. By further reducing these tolerances for the customer sets, we obtain the hybrid compression problem:

		\begin{problem}[Hybrid compression problem]
		\label{prob:hybrid}
		\begin{equation}	
		\begin{array}{lll}
		\mathrm{min} & \sum_{ij} e^{\prime}_{ij} \\
		& &\\
		\text{s.t.}&   \sum_j \left(e^{\prime}_{ij} - e^{\prime}_{ji}\right) = v_i, \forall i \in N &\\ & \\
		&  e_{ij}^{\prime} = e_{ij}, \forall (i, j) \in E^C  & \\
		&  e_{ij}^{\prime} \in [0, \infty), \forall (i, j)  \in E^D & \\
		\end{array}
		\nonumber
		\end{equation}
		\end{problem}

		Last, by further restricting tolerances, we obtain the conservative compression problem:

		\begin{problem}[Conservative compression problem]
		\label{prob:conservative}
		\begin{equation}
		\begin{array}{lll}
		\mathrm{min} & \sum_{ij} e^{\prime}_{ij} \\
		& &\\
		\text{s.t.}&   \sum_j \left(e^{\prime}_{ij} - e^{\prime}_{ji}\right) = v_i, \forall i \in N &\\ & \\
		&  0 \leq e_{ij}^{\prime}  \leq e_{ij}, \forall (i, j) \in E & \\
		\end{array}
		\nonumber
		\end{equation}
		\end{problem}

		All problems can be interpreted as standard linear programs, which can be solved in numerous ways. We propose specific closed form solutions for the non-conservative compression problem. For the conservative and hybrid approaches, the general case where the network is not chain ordering proof, a global solution can be obtained via linear programming techniques. The results presented in this paper were obtained using a network simplex method. we refer the reader to \cite{ahuja1993network} for details on the simplex algorithm, its mathematical properties and the relative proofs.
\newpage
\section{Efficiency ratios: invariance under scale transformations\label{app-rescaling-ratios}}
We show that the both the excess ratio and the compression efficiency ratio for conservative compression are invariant to scale transformations.

		\begin{lemma}\label{lemma:scaletransformations} Let $G = (N, E)$ a market with associated exposure matrix $e_{ij}$, and $G(\alpha) = (N, E(\alpha))$ a market with exposure matrix $e_{ij}(\alpha) = \alpha \times e_{ij}$, where $\alpha$ is a strictly positive real number. The following relations hold:
		\begin{enumerate}
		\item \label{longproof:start} $v_i^{\text{net}}(\alpha) = \alpha v_i^{\text{net}} \;\; \forall i \in V$;
		\item \label{longproof:first} $x(\alpha) = \alpha x$, where $x = \sum_{ij} e_{ij} \text{ and } x(\alpha) = \sum_{ij}e_{ij}(\alpha)$;
		\item \label{longproof:second} $m(\alpha) = \alpha m$
			\item \label{longproof:third}	$\Delta(G(\alpha)) = \alpha \Delta(G)$;
			\item \label{longproof:fourth} $\epsilon_n(G(\alpha)) = \epsilon_n (G)$;
			\item \label{longproof:fifth} $\rho_c(G(\alpha)) = \rho_c(G)$.

		\end{enumerate}
		\end{lemma}
			
\begin{proof}
Point \ref{longproof:start} holds since $$v_i^{\text{net}}(\alpha) = \sum_j \alpha e_{ij} - \sum_j \alpha e_{ji} = \alpha v_i^{\text{net}},$$
which implies that each net position is simply rescaled by a factor $\alpha$.
Points \ref{longproof:first} and \ref{longproof:second} are easily proven by multiplying by $\alpha$ and hence \ref{longproof:third} and \ref{longproof:fourth} follow straighforwardly by the definition of excess. 

For point \ref{longproof:fifth}, we exploit the programming characterisation of the conservative compression problem and show that the optimal solutions of the programme for $G(\alpha)$ coincides with that of $G$ rescaled by $\alpha$. 

The programme for $G(\alpha)$ can be expressed as follows:
	\begin{equation}
		\begin{array}{lll}
		\mathrm{min} & \frac{1}{\alpha} \sum_{ij} e^\prime_{ij}(\alpha) \\
		& &\\
		\text{s.t.}&   \frac{1}{\alpha} \sum_j \left(e^\prime_{ij} (\alpha)- e^\prime_{ji}(\alpha)\right) = \frac{1}{\alpha} v_i(\alpha) = \frac{1}{\alpha}\alpha v_i^{\text{net}}, \forall i \in N &\\ & \\
		&  0 \leq \frac{1}{\alpha} e_{ij}^\prime (\alpha) \leq  \frac{1}{\alpha} e_{ij}(\alpha), \forall (i, j) \in E & \\
		\end{array}
		\nonumber
		\end{equation}

By posing $e_{ij}^\prime (\alpha) = \alpha e_{ij}^\star$ we observe that $e_{ij}^\star = e_{ij}^\prime$. Point \ref{longproof:fifth} follows by computing the ratio $\rho_c(G(\alpha))$ and applying  \ref{longproof:third}.

\end{proof}

\newpage
\section{Sampling statistics\label{app-sampling-stats}}
		Table \ref{tab:generalstats-raw} reports the main statistics of the sampled data over time. The total notional of the selected 100 entities varies between 380Bn Euros and 480Bn Euros retaining roughly $30-34\%$ of the original total gross notional. The average number of counterparties across the 100 entities is stable and varies between 45 and 58 individual counterparties. 

		\begin{table}[ht]
		\centering
		\begin{tabular}{lccc}
		Time & \begin{tabular}{l}Gross notional\\ of 100 top ref.\\ (E+11 euros) \end{tabular}& \begin{tabular}{l} Share of \\ gross notional \\ of 100 top ref. \end{tabular} & \begin{tabular}{l} Avg \\participants \\ per ref.\end{tabular}\\ 
		\toprule
			Oct-14 	&	3.88	&	0.358	&	 54\\
			Nov-14 	&	4.16	&	0.349	&	 55  \\ 
			Dec-14 	&	4.4		&	0.357	&	 58  \\ 
			Jan-15 	&	4.73	&	0.361	&	 57 \\ 
			Feb-15 	&	4.67	&	0.355	&	 57  \\ 
			Mar-15 	&	4.35	&	0.351	&	 51  \\ 
			Apr-15 	&	3.87	&	0.338	&	 46 \\ 
			May-15 	&	3.91	&	0.337	&	 45 \\ 
			Jun-15 	&	3.86	&	0.343	&	 47  \\ 
			Jul-15 	&	3.9		&	0.347	&	 50  \\ 
		  	Aug-15 	&	3.9		&	0.344	&	 52  \\ 
			Sep-15 	&	3.94	&	0.350	&	 53  \\ 
			Oct-15 	&	4.08	&	0.349	&	 55  \\ 
			Nov-15 	&	4.18	&	0.351	&	 55  \\ 
			Dec-15 	&	4.24	&	0.348	&	 55 \\ 
			Jan-16 	&	4.39	&	0.351	&	 55  \\ 
			Feb-16 	&	4.33	&	0.348	&	 56  \\ 
			Mar-16 	&	3.94	&	0.350	&	 49 \\ 
			Apr-16 	&	4.37	&	0.352	&	 49 \\ 
		\bottomrule
		\end{tabular}
		\caption{General coverage statistics of the dataset over time: total outstanding gross notional of the sampled markets, share of sampled market's gross notional against the full dataset and average number of participant in each sampled market.}\label{tab:generalstats-raw}
		\end{table}

\newpage
\section{Excess and efficiency in bilaterally compressed markets\label{app-bilateral-market-compression-analysis}}
		In derivatives market like CDS markets, participants, specially dealers, reduce some positions by writing a symmetric contract in the opposite direction with the same counterparty. Analyzing the bilaterally compressed market thus allows us to quantify excess and compression efficiency beyond the redundancy incurred by this specific behavior.

		As we have seen, bilateral excess, on average, accounts for half the excess of the original markets. In order to understand excess and compression beyond bilateral offsetting, we analyze further the bilaterally compressed markets. First, we obtain dealer-customer network characteristics reported in Table~\ref{tab:bowtiestats-bil} after bilateral compression. While the participant-based statistics mirror Table~\ref{tab:bowtiestats-raw}, there is a reduction in all obligation-related statistics except the intra-customer density which remains the same: the average number of obligations is reduced by $25$ percentage points. while the intra-dealer share of notional is only affected by 5 percentage points. Hence, we see that, despite the density reduction, the bulk of the activity remains in the intra-dealer activity after bilateral compression.\footnote{Note that the average intra-customer density is equal to Table~\ref{tab:bowtiestats-raw}. In theory, we should have doubled the value as the density of the bilaterally netted intra-customer segment should be seen as the density of a undirected graph. We kept the previous definition to highlight the fact that the intra-customers obligations are not affected by the bilateral compression and avoid a misinterpretation of density increase.}

		\begin{sidewaystable}
			\centering
			\begin{tabular}{lccccccccc}

			Time & \begin{tabular}{c}Avg num.\\ dealers\end{tabular} & \begin{tabular}{c}Avg num. \\ customers \\buying\end{tabular} & \begin{tabular}{c}Avg num.\\ customers \\selling\end{tabular} & \begin{tabular}{c}Avg num.\\ obligations\end{tabular} &
			\begin{tabular}{c}Avg. share\\ intra-dealer \\ notional \end{tabular} & \begin{tabular}{c}Avg. \\ density\end{tabular} & \begin{tabular}{c}Avg.\\ intra-dealer \\ density\end{tabular} & \begin{tabular}{c}Avg.\\ intra-customer \\ density\end{tabular} \\ 
			  \toprule
			Oct-14	&		18 &              16 &               20 &             115.87 &                              0.767 &              0.075 &                    0.221 &                      0.0010 \\
			Nov-14	&	19 &              16 &               21 &             121.61 &                              0.779 &              0.077 &                    0.227 &                      0.0006 \\
			Dec-14	&	19 &              17 &               21 &             128.24 &                              0.777 &              0.076 &                    0.221 &                      0.0005 \\
			Jan-15	&	19 &              17 &               21 &             127.94 &                              0.778 &              0.075 &                    0.219 &                      0.0006 \\
			Feb-15	&	19 &              17 &               21 &             126.31 &                              0.782 &              0.075 &                    0.221 &                      0.0004 \\
			Mar-15	&	18 &              15 &               17 &             114.74 &                              0.786 &              0.078 &                    0.225 &                      0.0007 \\
			Apr-15	&	18 &              13 &               15 &             106.28 &                              0.786 &              0.079 &                    0.229 &                      0.0005 \\
			May-15	&	18 &              12 &               15 &             106.18 &                              0.782 &              0.078 &                    0.224 &                      0.0008 \\
			Jun-15	&	18 &              13 &               14 &             105.20 &                              0.783 &              0.076 &                    0.216 &                      0.0010 \\
			Jul-15	&	19 &              14 &               14 &             107.05 &                              0.766 &              0.072 &                    0.211 &                      0.0009 \\
			Aug-15	&	19 &              15 &               17 &             111.49 &                              0.776 &              0.073 &                    0.208 &                      0.0011 \\
			Sep-15	&	19 &              16 &               17 &             114.12 &                              0.755 &              0.071 &                    0.204 &                      0.0018 \\
			Oct-15	&	19 &              16 &               18 &             117.22 &                              0.766 &              0.072 &                    0.200 &                      0.0013 \\
			Nov-15	&	19 &              17 &               19 &             120.52 &                              0.762 &              0.072 &                    0.198 &                      0.0017 \\
			Dec-15	&	19 &              17 &               19 &             120.76 &                              0.772 &              0.072 &                    0.198 &                      0.0012 \\
			Jan-16	&	20 &              17 &               18 &             121.06 &                              0.774 &              0.072 &                    0.198 &                      0.0013 \\
			Feb-16	&	19 &              17 &               18 &             121.05 &                              0.763 &              0.071 &                    0.197 &                      0.0012 \\
			Mar-16	&	18 &              14 &               17 &             108.03 &                              0.739 &              0.070 &                    0.205 &                      0.0018 \\
			Apr-16	&	19 &              14 &               17 &             109.29 &                              0.759 &              0.071 &                    0.204 &                      0.0019 \\

			   \bottomrule
			\end{tabular}
			\caption{Dealers/customers statistics after bilateral compression.}\label{tab:bowtiestats-bil}
		\end{sidewaystable}

		In terms of excess, Table~\ref{tab:excess-stats-bil} complements the results from the bilateral compression efficiency and reports statistics similar to Table~\ref{tab:excess-stats-raw}.\footnote{The relationship between the bilateral compression efficiency, $\rho_b$ and the relative excesses in the original market, $\epsilon^o$, and the bilaterally compressed market, $\epsilon^b$, is given by $\rho_b=(1-\frac{1-\epsilon^{o}}{1-\epsilon^{b}})\frac{1}{\epsilon^{o}}$. This expression directly follows from the definition of each parameter.} At the extremes, we note again high degrees of variability: for example, in mid-January 2016, the minimum level of excess was $0.261$ while the maximum was $0.809$. Nevertheless, results on the means and medians are stable over time and alway higher than $0.5$. We thus see that, in general, around half of the gross notional of bilaterally compressed market remains in excess vis-a-vis market participants' net position. Note that the gross notional used here is the total notional left after bilateral compression on the original market.

		\begin{table}
			\centering
			\begin{tabular}{lrrrrrrr}
				\textbf{Total Excess}  & Oct-14 & Jan-15 & Apr-15 & Jul-15 & Oct-15 & Jan-16 & Apr-16 \\ 
				\toprule
				min          &  0.422 &  0.423 &  0.290 &  0.257 &  0.366 &  0.261 &  0.293 \\
				max          &  0.811 &  0.811 &  0.798 &  0.809 &  0.820 &  0.809 &  0.781 \\
				mean         &  0.614 &  0.621 &  0.614 &  0.602 &  0.597 &  0.570 &  0.558 \\
				stdev        &  0.087 &  0.087 &  0.091 &  0.095 &  0.097 &  0.112 &  0.098 \\
				first quart. &  0.562 &  0.558 &  0.562 &  0.544 &  0.531 &  0.489 &  0.503 \\
				median       &  0.617 &  0.618 &  0.614 &  0.613 &  0.594 &  0.569 &  0.566 \\
				third quart. &  0.670 &  0.684 &  0.674 &  0.663 &  0.654 &  0.653 &  0.635 \\
			\end{tabular}
			\caption{Excess statistics after bilateral compression}
			\label{tab:excess-stats-bil}
		\end{table}

		Table~\ref{tab:compression-results-bil} reports the results related to the efficiency of conservative and hybrid compression applied to the already bilaterally compressed market. On the extremes, both the conservative and the hybrid compression perform with various degrees of efficiency: the minimum amount of excess reduction via conservative compression (resp. hybrid compression) oscillates around $15\%$ (resp. $35\%$) while the maximum amount of excess oscillates around $90\%$ (resp. $97\%$). This shows that compression can perform very efficiently and very poorly with both approaches. However, the fact that conservative compression reaches $90\%$ of excess removal shows the possibility of having very efficient compression despite restrictive compression tolerances. The mean and the median of both approaches are stable over time: both around $60\%$ for the conservative compression and $75\%$ for the hybrid compression. Overall, we find that each compression algorithm is able to remove more than half of the excess from the market. The hybrid compression allows for greater performances as a result of relaxing intra-dealer compression tolerances.

		\begin{table}
			\centering
			\begin{tabular}{lrrrrrrr}
				\textbf{Conservative ($\rho_c$)}  & Oct-14 & Jan-15 & Apr-15 & Jul-15 & Oct-15 & Jan-16 & Apr-16 \\ 
				\toprule
				min          &  0.160 &  0.203 &  0.140 &  0.163 &  0.165 &  0.119 &  0.098 \\
				max          &  0.894 &  0.927 &  0.923 &  0.878 &  0.912 &  0.911 &  0.878 \\
				mean         &  0.568 &  0.622 &  0.599 &  0.592 &  0.555 &  0.552 &  0.525 \\
				stdev        &  0.166 &  0.160 &  0.164 &  0.158 &  0.175 &  0.183 &  0.172 \\
				first quart. &  0.456 &  0.505 &  0.512 &  0.489 &  0.435 &  0.437 &  0.409 \\
				median       &  0.562 &  0.636 &  0.594 &  0.591 &  0.537 &  0.550 &  0.546 \\
				third quart. &  0.685 &  0.729 &  0.728 &  0.705 &  0.680 &  0.687 &  0.643 \\
				\\
				\textbf{Hybrid ($\rho_h$)}  & Oct-14 & Jan-15 & Apr-15 & Jul-15 & Oct-15 & Jan-16 & Apr-16 \\ 
				\toprule
				min          &  0.370 &  0.460 &  0.377 &  0.281 &  0.259 &  0.281 &  0.135 \\
				max          &  0.971 &  0.973 &  0.968 &  0.963 &  0.977 &  0.974 &  0.981 \\
				mean         &  0.724 &  0.763 &  0.760 &  0.755 &  0.738 &  0.735 &  0.752 \\
				stdev        &  0.149 &  0.130 &  0.130 &  0.130 &  0.146 &  0.140 &  0.148 \\
				first quart. &  0.623 &  0.691 &  0.678 &  0.674 &  0.626 &  0.642 &  0.679 \\
				median       &  0.735 &  0.785 &  0.781 &  0.778 &  0.775 &  0.756 &  0.784 \\
				third quart. &  0.846 &  0.866 &  0.859 &  0.866 &  0.849 &  0.851 &  0.845 \\
			\end{tabular}
			\caption{Statistics of compression efficiency after bilateral compression}
			\label{tab:compression-results-bil}
		\end{table}

\end{document}